\documentclass[twocolumn,10pt]{IEEEtran}

\usepackage[multiple]{footmisc}
\input{def.tex}
\usepackage{dsfont}
\usepackage{minibox}
\DeclareSymbolFont{matha}{OML}{txmi}{m}{it}
\DeclareMathSymbol{\varv}{\mathord}{matha}{118}
\usepackage{eurosym}
\usepackage{hhline}
\usepackage{multicol}
\usepackage{eurosym}

\makeatletter

\IEEEoverridecommandlockouts
\begin{document}
\title{Towards Optimal Energy Efficiency in Cell-Free Massive MIMO Systems}
\author{A. Papazafeiropoulos, H. Q.  Ngo, P. Kourtessis,   S. Chatzinotas, and J. M. Senior \vspace{2mm} \\
%
\thanks{A. Papazafeiropoulos is with the Communications and Intelligent Systems Research Group,
University of Hertfordshire, Hatfield AL10 9AB, U. K., and with SnT at the University of Luxembourg, Luxembourg.  H. Q.  Ngo  is with the School of Electronics, Electrical Engineering and Computer Science, Queen's University Belfast, Belfast BT3 9DT, U.K. P. Kourtessis and John M. Senior are with the  Communications and Intelligent Systems Research Group, University of Hertfordshire, Hatfield AL10 9AB, U. K. S. Chatzinotas is with the SnT at the University of Luxembourg, Luxembourg. 
E-mails: tapapazaf@gmail.com,  hien.ngo@qub.ac.uk, p.kourtessis@herts.ac.uk, symeon.chatzinotas@uni.lu.
}}
\maketitle
\vspace{-12 mm}

 \begin{abstract}
Motivated by the ever-growing demand for \emph{green} wireless communications and the advantages of \emph{cell-free} (CF) massive multiple-input multiple-output (mMIMO) systems, we focus on the design of their downlink (DL) for optimal \emph{energy efficiency} (EE). To address this fundamental topic, we assume that each access point (AP) is deployed with multiple antennas and serves multiple users on the same time-frequency resource while the APs are Poisson point process (PPP) distributed, which approaches realistically their opportunistic spatial randomness. Relied on tools from stochastic geometry, we derive a lower bound on the DL average achievable spectral efficiency (SE). Next, we consider a realistic power consumption model for CF mMIMO systems. These steps enable the formulation of a tractable optimization problem concerning the DL EE, which results in the analytical determination of the optimal pilot reuse factor, the AP density, and the number of AP antennas and users that maximize the EE. 
Hence, we provide useful design guidelines for CF mMIMO systems relating to fundamental system variables towards optimal EE. Among the results, we observe that 
an optimal pilot reuse factor and AP density exist, while larger values result in an increase of the interference, and subsequently, lower EE. Overall, it is shown that the CF mMIMO technology is a promising candidate for next-generation networks achieving simultaneously high SE and EE.
\end{abstract}


\section{Introduction}
The rapid development of wireless communication systems, by means of the fifth generation (5G) networks and beyond, aim at higher data rates with adequate quality of service (QoS) but  with the reduction of energy consumption being of primary concern~\cite{Andrews2014}. 
Obviously, achieving higher data rates with less power consumption might seem like contradictory goals~\cite{Chen2011}, but that is not necessarily the case. A promising solution to provide higher data rates is achieved by means of the so-called network densification, which, unfortunately,  stumbles at the major bottleneck of increasing interference resulting in higher power consumption~\cite{Nguyen2017}. Hence, the fundamental arising question is how to increase the network data rate while achieving optimal energy efficiency (EE) at the same time.  Although both academia and industry already  have focused on the EE of cellular networks in the past years~\cite{Buzzi2016}, existing architectures cannot face the increasing complexity of future networks towards many devices, many antennas, and many bands. For this reason,  new innovative architectures  are required to address the crucial demanding \emph{green} specifications and considerations in next-generation networks.

In the direction of network densification, a key 5G  technology  (in terms of the number of antennas per area unit), known as massive multiple-input multiple-output (mMIMO) systems, has emerged by providing $10 \times$ higher data rate with comparison to conventional cellular systems~\cite{Marzetta2010,Papazafeiropoulos2015a,massivemimobook,BHS18A}.  
Although mMIMO systems can effectively deal with interference, the achievable EE is limited by the large propagation losses that are typical in cellular networks. An interesting alternative is to distribute a large number of antennas over the coverage area and operate these antennas in a network MIMO manner~\cite{Bjornson2010,Gong2016}\footnote{Despite that network MIMO has attracted  a lot of interest in the last decade~\cite{Bjornson2010,Gong2016}, its implementation is not feasible for practical systems  due to its substantial backhaul overhead.}. A practical embodiment of network MIMO is the cell-free (CF) mMIMO concept described in~\cite{Ngo2017}.

CF mMIMO consist of a deployment of a large number of access points (APs) that are distributed over the coverage area to coherently serve a large number of users on the same time-frequency resource. According to \cite{Ngo2017},   as the number of APs increases, we manage to take advantage of the favorable propagation and channel hardening properties, and finally, achieve very large spectral efficiency (SE) with simplified signal processing needing less overhead. However, herein, it is crucial to mention that the attractive properties of channel hardening and favorable propagation do not hold under all conditions. In particular, despite~\cite{Ngo2017} that accounted for these properties for single-antenna APs, in~\cite{Chen2018}, it was proved the opposite. Fortunately, it was shown that channel hardening and favorable propagation appear in the case of multiple-antennas APs (at least $5-10$ antennas) or low path-loss.   As a result, CF mMIMO can combine the benefits of coordination and low overhead. Moreover, CF mMIMO is a promising architecture because by increasing the number of APs the path-losses are improved and the macro-diversity is enhanced~\cite{Ngo2017}, which means that the transmit powers can be reduced. Unfortunately, these gains from CF mMIMO are achieved by deploying more hardware, which in turn, may increase the power consumption. Notably, even though CF mMIMO systems come with plausible potentials, the study of this technology is limited as literature reveals \cite{Ngo2017,Ngo2018,Chen2018,Bashar2019,Buzzi2017a,Chen2018,Bashar2018,Parida2018,Bjoernson2019a,Alonzo2019,Nayebi2016,Interdonato2019,Interdonato2019a,Bjoernson2020,Papazafeiropoulos2020}. For example,
the authors in~\cite{Buzzi2017a} achieved better data rates by suggesting a user-centric approach of CF mMIMO systems, where the APs serve a group of  users instead of all of them. Another interesting study concerns \cite{Papazafeiropoulos2020}, where the locations of the APs are Poisson point process (PPP) distributed, and the coverage probabilty was derived\footnote{Note that PPP is the most popular and tractable point 	process   for describing the spatial distribution of network nodes, e.g., see \cite{Andrews2011} and relevant works. The consideration of other spatial distributions  such as the Matern Hard-core point process could be the topic of future research.}.

Along the line concerning energy consumption, there are many works that have studied the EE of massive  MIMO~\cite{Bjoernson2016,Ren2017,Pizzo2018}. Specifically, in~\cite{Bjoernson2016}, the optimal  uplink (UL) EE of cellular networks was obtained analytically and examined thoroughly by using tools  from stochastic geometry with  PPP distributed BSs, and in~\cite{Pizzo2018}, the same methodology was applied for a multislope path-loss model. Only a few prior works have examined the  EE in CF mMIMO systems which is of particular interest since they are more beneficial than network MIMO
~\cite{Ngo2018,Nguyen2017a,Alonzo2019,Bashar2019}.  In particular, in the case of CF mMIMO systems, the EE was investigated in~\cite{Ngo2018,Nguyen2017a,Bashar2019} while, in~\cite{Alonzo2019},  the EE was investigated under a user-centric approach at millimeter-wave frequencies, but these works did not obtain analytical expressions for the EE, and did not take spatial randomness into account.  

\subsection{Motivation}
Most existing works on 5G networks focus on the SE while they neglect the importance of EE which is decreased when interference increases.  Network  MIMO, mitigating interference by means of coordination, is practically unattainable due to excessively high complexity in terms of hardware and information overhead. Luckily, CF mMIMO systems are an embodiment of massive MIMO and network MIMO systems and emerge as a promising feasible solution regarding coordination with low overhead exploiting the favorable propagation and channel hardening properties as the number of each AP antennas increases. Hence, the study of EE of CF mMIMO systems is of pivotal interest. 
Despite some existing works on the numerical optimization of the  EE of CF mMIMO systems~\cite{Ngo2018,Nguyen2017a}, there is no previous work deriving the optimal system parameters in closed form.
Most importantly, existing works, except \cite{Parida2018,Papazafeiropoulos2020}, focus on simplified network topologies such as grid-based models, and they do not account for the realistic spatial randomness of the APs\footnote{In \cite{Parida2018},  the spatial randomness of the APs was considered. However, the  distribution of the APs was again idealized and neglected their irregularity  since it was assumed uniform, i.e., a binomial point process (BPP) was applied. Moreover, certain approximations were made that result in a not strict analysis with not reliable expressions. For example, it was made the assumption of the nearest AP  and it was considered the mean contribution from the rest of the APs. Regarding our recent work in \cite{Papazafeiropoulos2020}, it was  relied on the deterministic equivalent (DE) analysis to obtain the DE signal-to-interference-plus-noise ratio (SINR) for a large number of APs. Also, it focused on the derivation of  the coverage probability and achievable rate for a large APs number.
}.  Especially, as the number of APs increases according to the concept of CF mMIMO, 
they are deployed opportunistically which means high irregularity. Although previous works mentioned that  the APs are randomly located, they consider a fixed number of APs while their randomness is not utilized in the analysis, but only in the simulations. These observations suggest that the analytical derivation of the optimal realistic EE of CF mMIMO systems, where the APs are distributed according to a PPP, is of paramount importance. Also, in order to extract trustworthy results, a realistic power consumption
model is needed to take both the transmit power and other system parameters into account.

\subsection{Contribution}
The main contributions are summarized as follows.
\begin{itemize}
  \item Contrary to existing works~\cite{Ngo2018,Nguyen2017a}, which did not account for the spatial randomness of the APs, and thus, are quite idealized, we apply tools from stochastic geometry and assume that the APs are PPP located. 
  In addition, contrary to \cite{Papazafeiropoulos2020}, our analysis  relies on a finite number of APs, and the aim of  this work is the study of the EE.   Also, we differentiate from \cite{Parida2018} that assumed a BPP for the  APs which is again idealistic.
  
 \item We derive a lower bound of the downlink (DL) average achievable SE for a finite number of APs being PPP distributed and having multiple antennas. Furthermore, we present a realistic power consumption model, specialized in CF mMIMO systems.
 
 \item Contrary to the common definition of area SE (ASE) in cellular networks, we provide a novel definition, which is necessary for CF mMIMO systems, and in general, in architectures with coordinated multi-point joint transmission (CoMP-JT).
 
 \item We obtain the optimal EE  of CF mMIMO systems with PPP distributed multiple-antenna APs by means of an analytical expression enabling to derive the optimal values for fundamental system parameters such as the network size in terms of  numbers of AP antennas and serving users.
 
 \item We shed light on the impact of the main system parameters on the optimal EE.  The results are of high practical interest since the analysis accounts for finite and realistic systems dimensions. Specifically, we obtain the optimal reuse factor, the optimal AP density, and the optimal number of AP antennas and users in closed-forms. For the sake of comparison, we also present results for a corresponding ``cellular'' mMIMO system and a small-cells (SCs) network. 
 \end{itemize}

\subsection{Paper Outline}  
The remainder of this paper is organized as follows.  Section~\ref{System} presents the system model of a CF mMIMO system with  multiple antennas APs being PPP distributed.  Sections~\ref{estimation} and~\ref{downlink} provide the UL training and DL transmission phases, respectively. Section~\ref{EnergyEfficiencyAnalysis} provides the analysis regarding the EE while Section~\ref{EnergyEfficiencyMaximization} presents the optimization of the EE and obtains the optimal system parameters in closed form.
The numerical results are placed in Section~\ref{Numerical}, and Section~\ref{Conclusion} concludes the paper.

\subsection{Notation}Vectors and matrices are denoted by boldface lower and upper case symbols, respectively. The symbols $(\cdot)^\T$, $(\cdot)^\H$, and $\tr\!\left( {\cdot} \right)$ express the transpose,  Hermitian  transpose, and trace operators, respectively. The expectation  operator is denoted by $\EE\left[\cdot\right]$.   Also, $\bb \sim \cC\cN{(\b0,\mathbf{\Sigma})}$ represents a circularly symmetric complex Gaussian vector with {zero mean} and covariance matrix $\mathbf{\Sigma}$. Finally, the superscript $^\star$ is used to represent optimal values.

\section{System Model}\label{System}
We consider a CF mMIMO system with multiple antennas at the APs and we model the practical spatial randomness of APs by means of stochastic geometry. Specifically, we assume that the APs, each having  $N\ge 1$  antennas,  are distributed  in the two dimensional Euclidean plane with their locations following a homogeneous PPP $\Phi_{\mathrm{AP}}$ with intensity $\lambda_{\mathrm{AP}}$ $\left[\mathrm{AP}/\mathrm{km}^{2}\right] $. In a specific realization of the PPP  $\Phi_{\mathrm{AP}}$, the number of APs in any  region $ \mathcal{A} $ of size $S$ in  $ \mathrm{km}^{2} $, denoted by $M$,  is a Poisson random variable with mean value 
\begin{align}
 \EE\left[ M\right] =\lambda_{\mathrm{AP}} S.\label{meanValue} 
\end{align}

Following the network MIMO principle, all the APs serve simultaneously all the single-antenna users on the same time-frequency resource{\footnote{Given that our focus is the study of CF mMIMO systems under practical assumptions,  the optimization of their EE by accounting for user-centric and scalable requirements  as in \cite{Buzzi2017a} and \cite{Bjoernson2019a}  is a topic of future research.}
Interestingly, the total number of antennas in $\mathcal{A}$ in a realization of the spatial process, denoted by $\mathcal{W}=MN$ is a Poisson random variable with mean $\EE\left[ \mathcal{W}\right]=N \lambda_{\mathrm{AP}} S$. 
We let $K$ denote the number of users in any given network realization. Their number is fixed and the users are selected at random from a  large set based on  round-robin scheduling\footnote{This choice is equivalent to a random user selection in each time-frequency resource block and  it is  a common assumption in the literature for analytical tractability. The study of the impact of optimal scheduling  is an interesting topic for future research.}.
 Notably, the number of users is  an optimization variable  while their locations are  uniformly distributed~\cite{Bjoernson2016}.
To consider a CF mMIMO scenario, the densities are chosen in order to fulfill the condition $\mathcal{W}\gg K$  in most realizations~\cite{Chen2018}. 

All APs are connected via a perfect fronthaul network to a central processing unit
\footnote{Although the fronthaul links are not perfect in practice, but degraded due to several reasons such as the quantization noise~\cite{Marsch2011,Bashar2018,Parida2018}, this work 	assumes perfect fronthaul connections to focus on the impact of a realistic spatial randomness of the APs. The consideration of the 	fronthaul links limitations is  of practical interest and is left for future work.}. Taking advantage of Slivnyak's theorem, we focus on a typical user, selected at random among the  users and indexed by $k$, in order to analyze the network performance~\cite{Chiu2013a}. For the ease of exposition, we assume that the typical user is located at the origin.

\subsection{Channel Model}\label{ChannelModel} 
In a realization of the PPP  $\Phi_{\mathrm{AP}}$, i.e., given $M$, let the $N \times 1$ channel  vector  $\bh_{mk}$ between the $m$th AP and the typical user be given by
\begin{align}
\bh_{mk}=l_{mk}^{1/2}\bg_{mk}, ~~~~~~m=1,\ldots,M~\mathrm{and}~k=1,\ldots,K
\end{align}
where $l_{mk}=\min\left( 1, r_{mk}^{-\al}\right)$ and $\bg_{mk}$ represent  independent path-loss and small-scale fading between the $m$th $N$ antenna AP and the typical user. In particular, the path-loss is described by means of a non-singular  bounded model with $\al>0$ being the path-loss exponent and $r_{mk}$ being the distance  between the $m$th AP and the $k$th  user~\cite{haenggi2009interference}.
Note that this bounded path-loss model is practical also at short distances while  an  bounded path-loss model is not suitable for the study  of CF mMIMO systems with stochastic geometry because it might result in unrealistically high power gain if the APs arbitrarily close to the user~\cite{Chen2018}. Given that this work accounts for the spatial randomness of the APs, the following analysis is dependent on the selection of the path-loss model. Although the majority of CF mMIMO works such as \cite{Ngo2017,Ngo2018} have considered another path-loss model, herein, for the sake of clarity and simplicity, we have considered a famous bounded path-loss model that will result in tractable expressions. Note that the three-slope path-loss model would be too complicated for the analysis. Also, both models provide similar insights regarding the parameters of the system under study in this work. The same reasons have contributed to the wide acceptance of the bounded model  in many  scenarios  modeled in terms of  stochastic geometry \cite{haenggi2009interference}. In addition, both types of distances, i.e., the distance between the $m$th AP located at $\bx_m$ in $\mathbb{R}^{2}$ and the typical user as well as the distances between the $m$th AP and the other users in $\mathcal{A}\backslash \{\bx_{m} \in \mathcal{A} \}$ 
follow the uniform distribution and are independent. Also, similar to other works on CF mMIMO systems, e.g.,~\cite{Ngo2017,Chen2018, Ngo2018,Nguyen2017a,Bashar2019}, we assume independent uncorrelated Rayleigh fading where the elements of $\bg_{mk}$ are independent and identically distributed (i.i.d.) $\mathcal{CN}\left( 0,1 \right)$ random variables. Note that this  assumption of uncorrelated channels is reasonable, since   the service antennas (APs) in CF  mMIMO systems are distributed over a large area and the AP antennas can be well separated. Hence, the set of scatterers is likely to be different for each AP and each user.

We consider a time-varying narrowband channel that is divided into coherence blocks, which are blocks of duration $T_{\mathrm{c}}$ in $\mathrm{s}$ and bandwidth $B_{\mathrm{c}}$ in $\mathrm{Hz}$ while the channels are fixed and frequency-flat. Each coherence block consists of $\tau_\mathrm{c}=B_{\mathrm{c}}T_{\mathrm{c}}$ samples (channel uses) and we follow the standard block fading model where independent channel realizations appear in every block \cite{massivemimobook}.
We employ the  time-division-duplex (TDD) protocol with an UL training phase of $\tau_{\mathrm{tr}}$ samples and two data transmission phases of $\tau_{\mathrm{d}}$ (DL) and $\tau_{\mathrm{up}}$ (UL) samples, respectively. Hence, we have $\tau_\mathrm{c}=\tau_{\mathrm{tr}}+\tau_{\mathrm{up}}+\tau_{\mathrm{d}}$ while the communication strategy is illustrated in Fig.~\ref{Fig0}.  In this work, we focus on the  UL training and DL data transmission phases. The duration of the latter can be expressed by $\tau_{\mathrm{d}}=\xi\left(\tau_\mathrm{c}-  \tau_{\mathrm{tr}}\right)$ with $\xi \le 1$, where $ \xi $ expresses the DL payload fraction transmission~\cite{Pizzo2018}.

 \begin{figure}[!h]
	\begin{center}
		\includegraphics[width=0.5\linewidth]{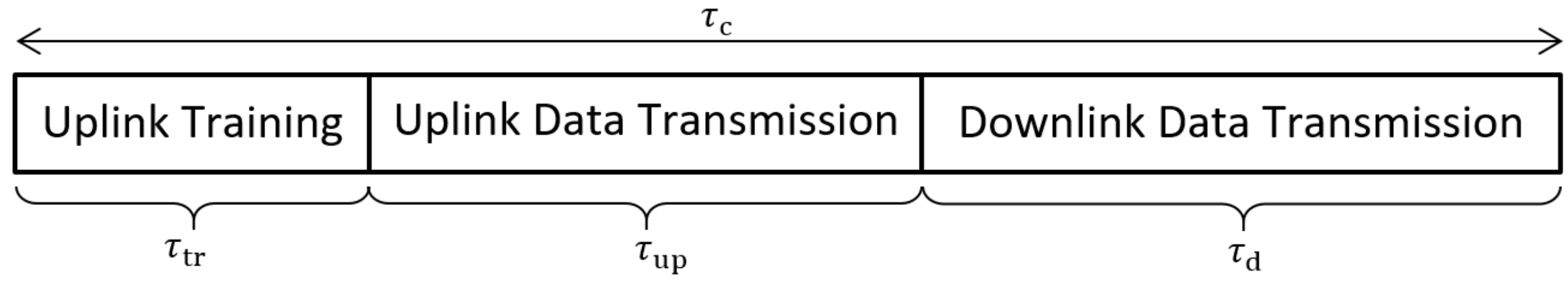}\vskip -5mm
		\caption{\footnotesize{The TDD transmission strategy.}}
		\label{Fig0}
	\end{center}
\end{figure}

\section{UL Channel Estimation}\label{estimation}
The construction of the precoder for the DL transmission requires the channel state information, which is obtained from the UL training phase. The $\tau_{\mathrm{tr}}$ channel uses for UL training need to be shared among all the users and there is room for $\tau_{\mathrm{tr}}$ mutually orthogonal pilot sequences. Since $K \gg \tau_{\mathrm{tr}}$ in most cases of interest, there will be pilot contamination. By introducing the reuse factor $\zeta=K/\tau_{\mathrm{tr}}$, we note that $ \zeta $ users share the same pilot sequences. 

In the training phase of one realization of the network, the $k$th user transmits a normalized pilot sequence  $\bpsi_{k}\in \mathbb{C}^{ \tau_{\mathrm{tr}}\times 1}$  with  $\|\bpsi_{k}\|^{2}=1$, and the received  $N \times \tau_{\mathrm{tr}}$ channel vector by the $m$th AP  is given by
\begin{align}
\!\!\!\tilde{\bY}_{m}^{\mathrm{tr}}&
\!= \! \sum_{i=1}^{K}\sqrt{\tau_{\mathrm{tr}} \rho_{\mathrm{tr}}}l_{mi}^{1/2}\bg_{mi}\bpsi_{i}^{\H}\!+\!\bn_{m}^{\mathrm{tr}},\label{eq:Ypt2}
\end{align}
where  $\rho_{\mathrm{tr}}$ is the  average transmit power while $\bn_{m}^{\mathrm{tr}}$ is the $N\times \tau_{\mathrm{tr}}$  additive noise vector at the $m$th AP consisted of i.i.d. $\mathcal{CN}\left( 0,1\right)$ random variables. In other words,  $\rho_{\mathrm{tr}}$  is actually the normalized signal-to-noise ratio (SNR). By projecting $\tilde{\by}_{mk}^{\mathrm{tr}}$ onto $\frac{1}{\sqrt{\tau_{\mathrm{tr}} \rho_{\mathrm{tr}}}}\bpsi_{k}$, we obtain
\begin{align}
\tilde{\by}_{mk}=&  \bg_{mk}l_{mk}^{1/2}\!+\! \sum_{i\ne k}^{K}l_{mi}^{1/2}\bg_{mi}\bpsi_{i}^{\H}\bpsi_{k}\!+\!\frac{1}{\sqrt{\tau_{\mathrm{tr}} \rho_{\mathrm{tr}}}}\bn_{m}^{\mathrm{tr}}\bpsi_{k}.\label{eq:Ypt3}
\end{align}

 With the assumption that the channel and   distances   statistics  are known a priori and that $\bpsi_{i}^{\H}\bpsi_{k} \in \{0, 1\}$ for all $i,k$, the $m$th AP obtains the linear minimum mean-squared error (MMSE) estimate according to \cite{Verdu1998}, i.e., $ \hat{\bh}_{mk}\!=\!{\mathrm{E}\!\left[\bh_{mk}\tilde{\by}_{mk}^{\H}  \right]}{\mathrm{E}^{-1}\!\left[\tilde{\by}_{mk}\tilde{\by}_{mk}^{\H}\right]}\tilde{\by}_{mk}$. Thus, we have
\begin{align}
 \hat{\bh}_{mk}
 &=\frac{{l_{mk}}}{\sum_{i=1}^{K}|\bpsi_{i}\bpsi_{k}^{\H}|^{2}l_{mi}+\frac{1}{{{{\tau_{\mathrm{tr}} \rho_{\mathrm{tr}}}}}}}\tilde{\by}_{mk}.\label{estimatedChannel1} 
\end{align}


The  estimation error vector $\tilde{\bee}_{mk}={\bh}_{mk}-\hat{\bh}_{mk}$ is independent of  $\hat{\bh}_{mk}$. Moreover, it follows that $\bh_{mk}\in\mathbb{C}^{{N}\times 1}\sim\mathcal{CN}\left( \b0, l_{mk}\Id_{N}\right)$, $\hat{\bh}_{mk}\in\mathbb{C}^{{N}\times 1}\sim\mathcal{CN}\left( \b0, \sigma_{mk}^{2}\Id_{N}\right)$ and $\tilde{\bee}_{k}\in\mathbb{C}^{{N}\times 1}\sim\mathcal{CN}\left( \b0,\tilde{\sigma}_{mk}^{2}\Id_{N}\right)$, where $ \sigma_{mk}^{2}=\frac{l_{mk}^{2}}{ d_{m}}$ and  $\tilde{\sigma}_{mk}^{2}\!=\!l_{mk}\left( 1-\frac{l_{mk}}{ d_{m}} \right)$ with $d_{m}\!=\!\left( \sum_{i=1}^{K}|\!\bpsi_{i}^{\H}\bpsi_{k}|^{2}l_{mi}\!+\!\frac{1}{{\tau_{\mathrm{tr}} \rho_{\mathrm{tr}}}} \right)$. 

\section{DL Transmission}\label{downlink} 
We now consider the DL transmission in one realization of the network, where the APs have multiple antennas and are PPP distributed. The goal is to derive the achievable SE with conjugate beamforming, taking into account the effect of pilot contamination as well as the spatial randomness of the APs.  "The choice of conjugate beamforming relies on its indication for  distributed architectures  due to no need for CSI exchange among the APs and the central unit~\cite{Ngo2017}. Also, 
conjugate beamforming  performs well in both CF mMIMO systems and SCs, and it provides the derivation of closed-form tractable expressions\footnote{Notably, the impact of  zero-forcing (zero-forcing) in terms of analytical closed-form results is the topic of ongoing research. Therein, the more robust regularized ZF, which  is indicated for better performance, is also studied but in terms of Monte-Carlo (MC) simulations since it does not provide closed-form expressions.}.

Note that although the choice of zero-forcing (zero-forcing) or regularized ZF  is indicated for better performance, their applications would not allow the derivation of any tractable closed-form expressions, which is one of the main contributions of this work.

The received signal by the typical user is given by
\begin{align}
 y^{\mathrm{d}}_{k}&=\sqrt{ \rho_{\mathrm{d}}}\sum_{i\in \Phi_{\mathrm{AP}}}\tilde{\bh}_{i}^{\H} \bs_{i}+z^{\mathrm{d}}_{k}\label{signal} \\
&=\sqrt{ \rho_{\mathrm{d}}} \sum_{m=1}^{M}\bh_{mk}^{\H}  \bs_{m}+z^{\mathrm{d}}_{k}.\label{signal1}
\end{align}

In~\eqref{signal}, the vector $\tilde{\bh}_{i}$ describes the  channel between the $i$th AP located at $\bx_{i}\in \mathbb{R}^{2}$ and the typical user including  small-scale fading and path-loss, $\rho_{\mathrm{d}} >0$ denotes the corresponding transmit power,  while  $\bs_{i}$ is  the transmitted signal from the $i$th AP, and $z^{\mathrm{d}}_{k}\sim \mathcal{CN}\left( 0,1 \right)$ is the additive white Gaussian noise at the $k$th user.  Since a realization of the system includes $M$ APs, the signal model described by~\eqref{signal} can be written as in~\eqref{signal1}. Notably,  the number $M$ is a random variable changing in every spatial realization of the APs. In~\eqref{signal1}, the vector $\bh_{mk}$ expresses the channel between the $m$th AP and the typical user while  $\bs_{m}$ is the transmit signal from the $m$th  AP, which is written as
\begin{align}
\bs_{m}= \sum_{k=1}^{K} \sqrt{ \eta_{mk} }\bff_{mk}q_{k},\label{ZF}
\end{align}
where $q_{k}\in \mathbb{C}$ is  the normalized transmit data symbol for  user $k$ satisfying $\EE\left[ |q_{k}|^{2}\right]=1 $.
The vector $\bff_{mk}=\hat{\bh}_{mk}\in \mathbb{C}^{N}$ expresses the linear precoder.
Also, we denote $\eta_{mk}=\mu\sigma_{mk}^{-4}$ with $ \mu $  obtained by means of the constraint of the   transmit power   $\mathbb{E}\left[\frac{ \rho_{\mathrm{d}}}{K}\bs_m\bs_m^{\H}\right]=\rho_{\mathrm{d}}$. This selection regarding $ \eta_{mk} $ aims at easing the following algebraic manipulations. Actually, it corresponds to a statistical channel inversion power-control policy~\cite{Bjoernson2016}. It allows each AP to allocate more power to the most distant users and less power to the closest ones. Note that the scaling does not result in any loss in the performance since the parameter $\mu$ is changed accordingly. 

Henceforth, for the sake of algebraic manipulations, we   denote  $\bh_{k} = [\bh_{1k}^T \cdots \bh_{Mk}^T]\sim\mathcal{CN}\left( \b0, \bL_{k}\right)$, $\hat{\bh}_{k}= [\hat{\bh}_{1k}^T \cdots \hat{\bh}_{Mk}^T]\sim\mathcal{CN}\left( \b0, \bPhi_{k}\right)$ and $\tilde{\bee}_{k}\in\mathbb{C}^{\mathcal{W}\times 1}\sim\mathcal{CN}\left( \b0,\bL_{k}-\bPhi_{k}\right)$. The matrices $\bL_{k} \in\mathbb{C}^{\mathcal{W}\times \mathcal{W}}$,   $\bPhi_{k}= \bL_{k}^{2}\bD^{-1} \in\mathbb{C}^{\mathcal{W}\times \mathcal{W}}$,  and $\bD\in\mathbb{C}^{\mathcal{W}\times \mathcal{W}}$  are block diagonal matrices with  elements given by the matrices $\left[ \bL_{k}\right]_{ww}=l_{mk}\Id_{N}$,  $\left[ \bPhi_{k}\right] _{ww}=\sigma_{mk}^{2}\Id_{N}$, $\left[ \bD\right]_{ww}=d_{m}\Id_{N}$, and $\left[ \bD\right]_{ww}=d_{m}\Id_{N}$, respectivetly, for $w=1,\ldots,\mathcal{W}$ and $\mathcal{W}=MN$.  We also define $\bC_{k}=\bPhi_{k}^{-1}$ with  $\left[ \bC_{k}\right]_{ww}=c_{mk}\Id_{N}$, where $c_{mk}=\sigma_{mk}^{-2}$.

After substituting~\eqref{ZF} into~\eqref{signal1}, the received signal by the typical user is given by
\begin{align}
&y^{\mathrm{d}}_{k}
  \!=\!\sqrt{  \rho_{\mathrm{d}} } \bigg(\EE\!\left[  \sum_{m=1}^{M}\eta_{mk}^{1/2}
   \bh_{mk}^{\H}\hat{\bh}_{mk}\right]\!q_{k}\!+\!  \sum_{m=1}^{M}\eta_{mk}^{1/2}{\bh}_{mk}^{\H}\hat{\bh}_{mk}q_{k}\!\nn\\
   &-\! \EE\!\left[ \sum_{m=1}^{M}\eta_{mk}^{1/2} {\bh}_{mk}^{\H}\hat{\bh}_{mk}\right]q_{k}\!+\! \sum_{i\ne k}^{K}\sum_{m=1}^{M}\eta_{mi}^{1/2}{\bh}_{mk}^{\H}\hat{\bh}_{mi}q_{i}\bigg)\!+\! z^{\mathrm{d}}_{k},  \label{filteredsignal}
\end{align}
where we have written 
 \eqref{filteredsignal} similar to~\cite{Medard2000}, in order to derive the SINR based on the fact that the users do not have any knowledge of the  instantaneous CSI,  but they are aware of its statistics
\footnote{Although, in general, channel hardening does not appear in CF mMIMO systems with single-antenna APs according to~\cite{Chen2018} and~\eqref{filteredsignal}, we exploit its property because  the proposed model considers multi-antenna APs. In fact, in~\cite{Chen2018}, it was shown that if $N\ge 5$, channel hardening is met. Taking into account for this limitation, numerical results, provided in Section~\ref{Numerical}, reveal  that  values of interest regarding $N$ is more than $10$,  which corroborates our analysis taking advantage of channel hardening.}. Note that the second term in \eqref{filteredsignal} expresses the desired signal while the fourth term describes the multi-user interference. By applying the well-established bounding technique in~\cite{Medard2000},  we consider that~\eqref{filteredsignal} represents a single-input single-output (SISO) system, where  the APs treat the unknown terms as uncorrelated additive noise. Thus, we obtain the effective SINR of the DL transmission from all the multi-antenna  APs to the typical user, conditioned on the number of APs and their distances from the users, as
\begin{align}
\!\bar{\gamma}_{k}=\frac{\Big|\EE\left[
	\bh_{k}^{\H}\bC_{k}\hat{\bh}_{k}\right]\Big|^{2}}{\sum_{i=1}^{K}\!\EE\!\left[ \Big|
	\bh_{k}^{\H}\bC_{i}\hat{\bh}_{i}\Big|^{2}\right]- \Big|\EE\!\left[
	\bh_{k}^{\H}\bC_{k}\hat{\bh}_{k}\right]\!\!\Big|^{2} \!+\frac{1}{{\mu p_{\mathrm{d}}}}}.\label{SINR15} 
\end{align} 
Notably,  the matrices in~\eqref{SINR15} are random because they include the number of APs and the distances between the APs and the users , being  random variables changing in each realization.

\begin{proposition}\label{PropDetSINR} 
Given a realization of the network with $M$ APs and $K$ users, the effective SINR of the DL transmission 
at the typical user 
 in a CF mMIMO system, accounting for pilot contamination and  conjugate beamforming, is given by \eqref{SINRFinite}.
\begin{figure*}
\begin{align}
 \bar{\gamma}_{k}=\frac{M^{2}N}{\sum_{i=1}^{K} \tr \left(\bC_{i}\left(N \bL_{k}+\frac{1}{K p_{\mathrm{d}}}\Id_{M}\right)\right)+N\sum_{i \ne k}^{K}\tr ^{2}\left( \bL_{k}\bL_{i}^{-1}\right)-N\tr\left( \bD\bL_{k}^{-1}\right)+M}.\label{SINRFinite} 
\end{align}
	\hrulefill
\end{figure*}
\end{proposition}
\begin{proof}
See Appendix~\ref{SINRproofDistances}.
\end{proof}
\begin{remark}
The scaling in the numerator with $N$ corresponds to the array gain resulting from the coherent transmission of the $N$ antennas per AP. Moreover, the summations in the denominator are  over the number of users $K$ because as their number increases, the interference increases.
\end{remark}

\section{EE Analysis}\label{EnergyEfficiencyAnalysis}
In this section, we provide the definition  of the EE  of CF mMIMO systems where the APs locations follow a PPP distribution. 
Specifically, we focus on the  analytical derivation of the DL  EE  by first obtaining a lower bound on the average ASE, and then,  presenting a realistic power consumption model.   The power consumption is expected to increase rapidly with the number of APs, i.e., their density. Hence, it is of paramount importance to quantify the relevant efficiency of a CF mMIMO system. 
\begin{definition}
The EE, denoted by $ {\mathrm{EE}} $, expresses  the amount of reliably transmitted information  per unit of energy, which is defined mathematically as
\begin{align}
{\mathrm{EE}}\left[ \mathrm{bit}/\mathrm{Joule}\right]
 &=\frac{B_{\mathrm{w}}\left[ \mathrm{Hz}\right]  \cdot {\mathrm{ASE}}~\left[\mathrm{bit}/\mathrm{s}\Big/\!\!\left(\mathrm{Hz}\cdot \mathrm{km}^{2}\right)\right] }{\mathrm{APC\left[ \mathrm{W}/\mathrm{km}^{2}\right] }},\label{DefMaximization}
\end{align}
where $B_{\mathrm{w}} $, $ {\mathrm{ASE}} $, and APC describe the transmission bandwidth, the ASE, and the  area power consumption (APC), respectively.
\end{definition}

We continue with the derivations of  ASE and APC.

\subsection{Area Spectral Efficiency}\label{AreaSpectralEfficiency}
Taking advantage of the property of the typical user, stating that it is statistically equivalent with any other user in the network, the ASE is provided by
\begin{align}
{\mathrm{ASE}}=K^{'} \bar{R}~~~~\left[\mathrm{bit}/\mathrm{s}\Big/\!\!\left(\mathrm{Hz}\cdot \mathrm{km}^{2}\right)\right],\label{tse}
\end{align}
where $ K^{'}=K /S$ is the number of users per area $ S$,  and   $\bar{R}$ is the average SE per user.  We have  $\bar{R}=\bar{R}_{k}$, where $ \bar{R}_{k} $, provided below, is  the average DL SE of user $k$, being statistical equivalent with any other user in the network.

\begin{remark} 		Contrary to the common definition for the  ASE  in cellular systems with no cooperation \cite{Bjoernson2016,Pizzo2018}, the CF mMIMO architecture necessitates  a new definition.  		Specifically, in CF mMIMO systems each user receives joint transmission from multiple sources (APs), and the received SINR at the user is obtained from the sum of received signals from all these serving APs. Therefore, this received SINR is not the same as the received SINR computed in a single BS association network. Consequently, the  definition of ASE \cite{Bjoernson2016}, where the received user rate (i.e., per transmission link rate) is multiplied with the AP density does not hold in this scenario.  \end{remark}

Since the DL capacity for this network including imperfect CSI in not known, we follow the common approach, especially in the area of mMIMO~\cite{Marzetta2010,Bjoernson2016a},  focusing on the derivation of achievable lower bounds on the ergodic capacity. In particular, the following lemma provides a tractable lower bound on the ergodic capacity for any given realization of $\Phi_{\mathrm{AP}}$.
\begin{lemma}[\cite{Hoydis2013}]\label{lemma1} 
A lower bound on the DL ergodic channel capacity of the typical user $k$ in a CF mMIMO system with conjugate beamforming and PPP distributed APs   for any given realization of $\Phi_{\mathrm{AP}}$  is provided by
\begin{align}
 \bar{R}_{k}=\left( 1-\frac{ K}{ \zeta \tau_{c}} \right)\log_{2} \left( 1+ {\bar{\gamma}}_{k}\right) ~~~~~~\mathrm{[b/s/Hz]},\label{Ratebar} 
\end{align}
where  $K$ is the number of users,  $\zeta$ is the pilot reuse factor, and $\tau_{c}$ is the channel coherence interval in number of
samples while ${\bar{\gamma}}_{k}$ is given by~\eqref{SINRFinite}. 
\end{lemma}

The average SE per user is obtained by applying the expectation at~\eqref{Ratebar} over the APs locations. We resort to  Jensen's inequality to derive a closed-form lower bound for the DL achievable $R_{k}$ and avoid intractable lengthy numerical integral evaluations with respect to the APs distances.

\begin{theorem}\label{PropDetSINRDistances2} 
A lower bound on the DL average SE per user with conjugate beamforming precoding in a CF mMIMO system with  multi-antenna APs is obtained by
\begin{align}
 {R}_{k}=\left( 1-\frac{ K}{\zeta \tau_{c}} \right) \log_{2} \left( 1+ {\gamma}_{k}\right) ~~~~~~\mathrm{b/s/Hz},\label{Ratebar1} 
\end{align}
where  ${\gamma}_{k}=1/\check{\gamma}_{k}$ with $\check{\gamma}_{k}$ given by
\begin{align}
 \check{\gamma}_{k}\!&=\!
 \sum_{j=1}^{K}|\psi_{j}\psi_{k}^{\H}|^{2}\!\!\left( \!\frac{ \al \!-\!2  }{\al\pi N p_{\mathrm{d}} }\!+\!K\!-\!1\!\! \right)\!\nn\\
 &+\! \frac{\zeta}{{\al \pi K \rho_{\mathrm{tr}}}}\!\left(\!\! (K\!-\!1)\left( \al\!-\!2  \right)\!+\! \frac{\left( \al\!-1\!  \right) \!}{ N p_{\mathrm{d}} } \right)\!+\!\lambda_{\mathrm{AP}}\!\left( K\!-\!1 \right)\!.\label{checkGamma} 
\end{align}
\end{theorem}
\begin{proof}
See Appendix~\ref{SINRproofDistances2}.
\end{proof}

Notably, if we shed further light into \eqref{Ratebar1}, we observe that the sum  $ \mathrm{SE} $ is   a strictly quasi-concave function of
the number of users $ K $ while the optimal number of antennas per AP depends on the AP density and the quality of CSI in terms of $ N $ and $ \zeta $, respectively. 
These observations are in line with \cite{Papazafeiropoulos2020}, accounting also for the spatial AP randomness.

Although we have applied the law of large numbers regarding the number of APs during the derivation of this proof, it is known that this law is applicable and valid in the case of a finite number of APs obeying to $ M>8 $  \cite{Couillet2011}. Obviously, this range is of practical interest in CF mMIMO systems. The agreement of the analytical results with MC simulations in Section \ref{Numerical} for finite system dimensions confirms this assertion. Thus,  Theorem \ref{PropDetSINRDistances2} and the following results describe realistic systems of finite dimensions.

\subsection{Area Power Consumption}\label{AreaPowerConsumption}
The sources of the area power consumption of a CF mMIMO system  are the power usage during the  transmission $P_{\mathrm{TX}}$ and the circuitry of the system $P_{\mathrm{CPC}}$. Following a similar approach to~\cite{Bjoernson2015a,Pizzo2018} but specialized to CF massive systems, we  have
\begin{align}
\mathrm{APC}=\lambda_{\mathrm{AP}} \left(\frac{1}{\al_{\mathrm{eff}}}P_{\mathrm{TX}}+P_{\mathrm{CPC}}  \right)\!\!,\label{APC1} 
\end{align}
where $\al_{\mathrm{eff}}\!\in\! \left(\!\!\right.0,1\!\!\left.\right]$ is the  power amplifier efficiency. Note that $P_{\mathrm{TX}}$ concerns  both the average powers for the UL pilot and DL payload transmissions. Regarding $P_{\mathrm{CPC}}$, it describes the circuitry dissipation in terms of cooling, power supply, backhaul signaling, digital signal processing, etc. Although the majority of works assume that  $P_{\mathrm{CPC}}$ is a  fixed constant, this is not a realistic assumption, and obviously, not a good design methodology. In practice, each antenna is accompanied by dedicated circuits that contribute to the system power consumption. Above this, if $\mathrm{APC}$ was independent of $N$, the $\mathrm{ASE}$, increasing with $N$, would result in an unbounded EE as $N$ increases, which is irrational~\cite{Bjoernson2015a}. Hence, it is of dire necessity to incorporate in our EE analysis an accurate model for the power consumption. 

\begin{proposition}\label{PropAPC} 
 A  realistic model for the DL APC of CF mMIMO systems is given by
 \begin{align}
\mathrm{APC}\!\left( \btheta \right)&= \lambda_{\mathrm{AP}}( C_{0}+C_{1}K +C_{2}K^{2}+D_{0}N+D_{1}NK\nn\\
&-D_{2}NK^{2})+\mathcal{B} B_{\mathrm{w}} \mathrm{ASE}\!,\label{PropAPC12} 
   \end{align}
   where $C_{0}=P_{\mathrm{FP}}+P_{\mathrm{LO}}$,  $C_{1}=\frac{ B_{\mathrm{w} }}{7 L_{\mathrm{AP}}\tau_{\mathrm{c}}}-\frac{\xi \rho_{\mathrm{d}}}{\al_{\mathrm{eff}}\zeta \tau_{\mathrm{c}}}+P_{\mathrm{UE}}$, $C_{2}=\frac{ 1}{\al_{\mathrm{eff}}\zeta \rho_{\mathrm{tr}} \tau_{\mathrm{c}}}$, $D_{0}=P_{\mathrm{AP}}$, $D_{1}=\frac{3 B_{\mathrm{w}}}{L_{\mathrm{AP}}}+ \frac{3 B_{\mathrm{w}}}{L_{\mathrm{AP}}\tau_{\mathrm{c}}}$,   $D_{2}=\frac{3 B_{\mathrm{w}}\left( \xi-1 \right)} {L_{\mathrm{AP}}\zeta \tau_{\mathrm{c}}}$, and $\mathcal{B}= \left( P_{\mathrm{COD}} + P_{\mathrm{DEC}} +P_{\mathrm{BT}} \right)$.
\end{proposition}
\begin{proof}
 See Appendix~\ref{AreaPowerConsumptionLemmaproof}.
\end{proof}

It is worthwhile to mention that~\eqref{PropAPC12} is written in a polynomial structure that will facilitate the  optimization taking place in the following section.


\section{EE Maximization}\label{EnergyEfficiencyMaximization} 
This section elaborates on the main objective of this work, which is the maximization of the constrained EE with respect to the parameters defining the size of the network (e.g., the AP density and the number of users)
\footnote{The study of the EE optimization with regard to the transmit power in CF mMIMO systems with PPP distributed APs is the topic of our ongoing research.}. In other words, we scrutinize the tuple of system parameters $\btheta=\left(\zeta, \lambda_{\mathrm{AP}},K,N \right)$  obeying to the problem
\begin{align}\begin{split}
             \btheta^{\star}=\arg\max_{\btheta \in \Theta} \mathrm{EE}\!\left( \btheta \right)&=\frac{B_{\mathrm{w}}{\mathrm{ASE}}\!\left( \btheta \right)}{\mathrm{APC}\!\left( \btheta \right)}\\
 \mathrm{subject}~\mathrm{to}~\bar{\gamma}_{k}\!\left( \btheta \right)&=\gamma_{0},
              \end{split}\label{Maximization} 
\end{align}where ${\mathrm{ASE}}\!\left( \btheta \right)= K^{'} {R}$ with ${R}={R}_{k}$, where  $ {R}_{k} $ is given by Theorem~\ref{PropDetSINRDistances2}, $\mathrm{APC}\!\left( \btheta \right)$  is provided by Proposition~\ref{PropAPC}, $\bar{\gamma}_{k}$  is obtained by Theorem~\ref{PropDetSINRDistances2} while $\gamma_{0}>0$ is a design parameter. The set $\Theta$, including the feasible parameters values, is defined as $\Theta=\{\btheta : \lambda_{\mathrm{AP}}\ge 0,\zeta \ge 1, K/\zeta\le \tau_{\mathrm{c}}, \left( K,N \right)\in \mathbb{Z}_{+}\}$.  The constraint in~\eqref{Maximization} prevents from an optimal tuple of parameters with a low unacceptable achievable rate while it demands a specific QoS \cite{Bjoernson2016,Pizzo2018}. 

We aim at solving \eqref{Maximization} for either $ \lambda_{\mathrm{AP}}, N,K $ when the remaining parameters
are fixed. The advantage of this approach is to obtain closed-form expressions for 
the optimal EE and to shed light into 	the interplay among these parameters.
\subsection{Feasibility}
The optimization problem in~\eqref{Maximization} is feasible for a certain range of values of $\gamma_{0}$ because of the multiuser interference.  
\begin{lemma}
 The feasibility range of values of $\gamma_{0}$, used in the maximization problem for CF mMIMO systems~\eqref{Maximization},  can be described by
 \begin{align}
  \gamma_{0}<\frac{1}{\lambda_{\mathrm{AP}} }.
 \end{align}
\end{lemma}

\begin{proof}
 In order to obtain the range of values of $ \gamma_{0}$, we simplify the expression of the SINR, being the inverse of~\eqref{checkGamma} by noticing that it is a monotonically increasing function of $N$. Hence, deriving its upper limit as $ N\to \infty$\footnote{Although the assumption of an infinite number of antennas per AP, i.e., $ N \to \infty $ is impractical, it is used here just for showing the feasibility range  of the target SINR and has no impact on the main results since the following analysis is for finite $ N $.}, we obtain \eqref{proof1}.
 	\begin{figure*}
\begin{align}
  \lim_{N\to\infty} \bar{\gamma}_{k}\!=\!\frac{ \al \pi \rho_{\mathrm{tr}} K}{ \al \pi K\left(\sum_{j=1}^{K}|\psi_{j}\psi_{k}^{\H}|^{2}  \left( K\!-\!1 \right) \!+\!K \lambda_{\mathrm{AP}}  \right) \rho_{\mathrm{tr}}+\left( \al\!-\!2 \right) \left( K\!-\!1 \right)\zeta\!}.\label{proof1}
\end{align}
	\hrulefill
\end{figure*}
Since the upper limit is a decreasing function of the optimizable variable $\zeta$, we exploit the constraint $\zeta=K/\tau_{\mathrm{tr}}$ by taking its minimal value when $K=1$, and we obtain  the feasible $ \gamma_{0}$.\end{proof}

This lemma reveals that the upper limit of the SINR depends only on  $\lambda_{\mathrm{AP}} $ as $N \to\infty$. In the case of CF mMIMO, the typical value concerning the number of APs is$~100-200$~\cite{Ngo2017} which is equivalent to a density $\lambda_{\mathrm{AP}}\approx 10^{-4}~\mathrm{m}^{-2}$. In such case, e.g., $\lambda_{\mathrm{AP}}=10^{-4}~\mathrm{m}^{-2}$, the average SE per user is $\log_{2}\left( 1+100 \right)\approx 13.29~\mathrm{b/s/Hz}$. This value, showing the feasibility of the optimization problem described by~\eqref{Maximization}, is larger than the SE of currently applied systems~\cite{Holma2011}. Hence, the optimization problem under study is quite meaningful for practical systems. 
\subsection{Optimal Pilot Reuse Factor}
Herein, we derive  the optimal pilot reuse factor $\zeta^{\star}$ while the rest of the parameters are fixed.
\begin{theorem}\label{ReuseFactor} 
Let any set of $\{\lambda_{\mathrm{AP}}, K, N\}$  resulting in the feasibility of the maximization of EE given by~\eqref{Maximization}. The  optimal pilot reuse factor, satisfying the SINR constraint, is obtained by
\begin{align}
  \zeta^{\star}=\frac{\al \pi  K N   \rho_{\mathrm{tr}}  \rho_{\mathrm{d}}-\gamma_{0}Q_{1}}{\gamma_{0}Q_{2}}.\label{zetastar} 
\end{align}
\end{theorem}
\begin{proof}
 The reuse factor $\zeta^{\star}$ is obtained by  focusing on the constraint and  collecting the terms including $\zeta$ in the SINR given by~$\bar{\gamma}_{k}=1/\check{\gamma}_{k}$ as
 \begin{align}
 \gamma_{0}= \frac{\al \pi  N   \rho_{\mathrm{tr}}  \rho_{\mathrm{d}}}{  Q_{1}-\zeta Q_{2}},\label{proofpilotbeta} 
 \end{align}
 where  we set 
 \begin{align}
  Q_{1}&= K \rho_{\mathrm{tr}}\bigg( \!\! \left(\al-2  \right) \sum_{j=1}^{K}|\psi_{j}\psi_{k}^{\H}|^{2} \nn\\
  &+\al \pi N  \rho_{\mathrm{d}}\left( K-1 \right)\Big(  \sum_{j=1}^{K}|\psi_{j}\psi_{k}^{\H}|^{2} +\lambda_{\mathrm{AP}}\Big) \!\!\bigg)\!, \nn\\
  Q_{2}&=\left( \al-1 +N \rho_{\mathrm{d}} \left( \al-2 \right)\left( K-1 \right) \right)\!/K,
 \end{align}
 and we solve~\eqref{proofpilotbeta} with respect to $\zeta$.
\end{proof}

Theorem~\ref{ReuseFactor} provides the dependence of $\zeta^{\star}$ on the rest of the system parameters. According to its physical interpretation, a smaller pilot reuse factor, meaning a larger training phase, results in both    more precise channel estimation and less pilot contamination. Intuitively, a better channel estimation increases the SE, or equivalently, a better SINR constraint $ \gamma_{0}$ is allowed, which comes to agreement with~\eqref{zetastar}. It is shown that  $\zeta^{\star}$ is a decreasing function of $Q_{1}$ and $Q_{2}$, which both are increasing functions of $K$. However, a larger $K$ means higher interference, requiring a better channel estimation, i.e., a lower $\zeta$ which admits to the dependence shown by~\eqref{zetastar}.
\subsection{Optimal  APs Density}
After plugging~\eqref{zetastar} into  the optimization problem,~\eqref{Maximization} is written as
\begin{align}\begin{split}
            \mathrm{EE}\!\left(  \zeta^{\star}, K, N \right)&=\frac{B_{\mathrm{w}}{\mathrm{ASE}}\!\left( \zeta^{\star}, K, N\right)}{\mathrm{APC}\!\left(  \zeta^{\star}, K, N \right)}\\
 \mathrm{subject}~\mathrm{to}~1&\le\frac{\al \pi  N   \rho_{\mathrm{tr}}  \rho_{\mathrm{d}}-\gamma_{0}Q_{1}}{\gamma_{0}Q_{2}}\le \frac{K}{\tau_{\mathrm{c}}}.
              \end{split}\label{MaximizationZetastar} 
\end{align}
\begin{theorem}\label{OptimalLambda} 
 Let any set of $\{K,N\}$  keeping the optimization problem \eqref{MaximizationZetastar} feasible. For  fixed  $K$ and $N$, the EE  is maximized by 
 \begin{align}
 \lambda_{\mathrm{AP}}^{\star}=\min \left( \max\left( \lambda_{\mathrm{AP}_{0}},\lambda_{\mathrm{AP}_{1}} \right), \lambda_{\mathrm{AP}_{2}}\right),\label{optimallambda} 
  \end{align}
where 
\begin{align}
 \lambda_{\mathrm{AP}_{0}}=\frac{\left(a_{1}+a_{3}\right)G+\sqrt{a_{2} a_{3}a_{4}\left(a_{1}+a_{3}\right)G}}{a_{2}a_{4}G}\label{lambda0}
\end{align}
with
\begin{align}
G=a_{2}\left(a_{4}+a_{5}+a_{6}K\log\left(1+\gamma_{0}\right)\right)
\end{align}
while $ \lambda_{\mathrm{AP}_{1}} =\frac{a_{3}-a_{1}}{a_{2}} $,  $ \lambda_{\mathrm{AP}_{2}} =\frac{\tau_{\mathrm{c}}/\left(K a_{3}\right)+a_{1}}{a_{2}} $, and the parameters $\{a_{i}\}$ are provided in Table~\ref{ParameterValues3}.
\end{theorem}
\begin{proof}
 Both  $ {\mathrm{ASE}} $ and APC include the term $\zeta^{\star}\tau_{\mathrm{c}}/K$. Hence, we proceed with its computation which gives $\zeta^{\star}\tau_{\mathrm{c}}/K=\frac{a_{2}\lambda_{\mathrm{AP}}-a_{1}}{a_{3}}$. Then, after subistituting this term into the objective funtion of~\eqref{MaximizationZetastar}, the EE becomes
 \begin{align}
 \mathrm{EE}\!\left(\zeta^{\star} \right)=\frac{\frac{K \xi }{S\lambda_{\mathrm{AP}}}\left( 1- \bar{a}\right)\log_{2}\left( 1+\gamma\right)}{a_{4}+a_{5}\bar{a}+a_{6} K \xi \left( 1- \bar{a}\right)\log_{2}\left( 1+\gamma\right )},\label{proofN1} 
 \end{align}
 where $ \bar{a} =\frac{a_{3}}{a_{2}\lambda_{\mathrm{AP}}-a_{1}}$. Following the approach in \cite[Lem. 3]{Bjoernson2015a}, it can be shown that \eqref{proofN1} is a quasi-concave function  of $ \lambda_{\mathrm{AP}} $. Thus, \eqref{lambda0} is obtained by taking the first derivative of \eqref{proofN1} and equating to zero. Given that the constraint in \eqref{MaximizationZetastar} depends on $ \lambda_{\mathrm{AP}} $, we obtain  $ \lambda_{\mathrm{AP}_{1}} $ and $ \lambda_{\mathrm{AP}_{2}}$. 
\end{proof}

\begin{table*}
\caption{Optimization Parameters for Optimal  APs Density $\lambda_{\mathrm{AP}}$ and Number of Antennas $ N $}
\begin{center}
                                             \begin{tabulary}{\columnwidth}{ | c | c | |c|c|}\hline
\!\!\!{\bf Parameter} \!\!\!\!\!&{\bf Value}&\!\!\!{\bf Parameter}\!\!\! &{\bf Value}\\ \hline\hline
 $a_{1}$& $\!\!\!\rho_{\mathrm{tr}}\tau_{\mathrm{c}}K\!\left(\! \al \pi  N \!\rho_{\mathrm{d}}\!-\!\left(\! \left( \al\!-\!2 \right)\!+\!\al \pi N\! \rho_{\mathrm{d}} \right)\!\gamma\!  \sum_{j=1}^{K}\!\!|\psi_{j}\psi_{k}^{\H}|^{2} \! \right)$ \!\!\!\!\!\!&$b_{1}$&  $K \left( \al-1 \right)$\\ \hline
$a_{2}$ &  $\al \pi \gamma  \rho_{\mathrm{d}}\rho_{\mathrm{tr}}\tau_{\mathrm{c}}K N \left( K-1 \right)$&$b_{2}$ & $\rho_{\mathrm{d}} K \left( K-1 \right)\left( \al-2 \right) $\\ \hline
$a_{3}$&  $\gamma K N \rho_{\mathrm{d}}\left(\al-1+ \left( K-1 \right)\left( \al-2 \right) \right)$&$b_{3}$&\!\!  \!\!\!\!$\al \pi \rho_{\mathrm{d}}\rho_{\mathrm{tr}}\tau_{\mathrm{c}}\!\left(\! 1\!-\!\gamma \!\left(  \sum_{j=1}^{K}\!\!|\psi_{j}\psi_{k}^{\H}|^{2}\!+\!\lambda_{\mathrm{AP}}\! \!\right)\!\!\left( K\!-\!1\! \right)\!\! \right)$\!\!\!\!\!\\ \hline
$a_{4}$ &  ${C}_{0}+C_{11}K +N\left( D_{0}+D_{1}K \right)$&$b_{4}$ &  $ \rho_{\mathrm{tr}}\tau_{\mathrm{c}}\gamma \left( \al-2 \right)\sum_{j=1}^{K}|\psi_{j}\psi_{k}^{\H}|^{2}$ \\ \hline
 $a_{5}$ &  $ \frac{1}{\al_{\mathrm{eff}}}\left( \left( \xi-1 \right)\rho_{\mathrm{d}} +K \rho_{\mathrm{tr}} \right)-\frac{3BKN \xi}{L_{\mathrm{AP}}}$&$b_{5}$ &  $\left({C}_{0}+C_{11}K\right)/\lambda_{\mathrm{AP}}$ \\ \hline
 $a_{6}$ &  $P_{\mathrm{COD}} + P_{\mathrm{DEC}} +P_{\mathrm{BT}}$&$b_{6}$ &  $\left(D_{0}+D_{1}K -D_{2}K^{2}\right)/\lambda_{\mathrm{AP}}$ \\ \hline
 $a_{7}$&---&$b_{7}$ &  $\left( c_{22}-d_{22} -c_{12}\right)K^{2}/\lambda_{\mathrm{AP}}$ \\ \hline
 $a_{8}$&---&$b_{8}$ &  $\left(P_{\mathrm{COD}} + P_{\mathrm{DEC}} +P_{\mathrm{BT}}\right)/\lambda_{\mathrm{AP}}$ \\ \hline
  \end{tabulary}\label{ParameterValues3}
                                         \end{center}
                                                                          \vskip -8mm
                                        \end{table*}


\subsection{Optimal Number of AP Antennas and Users}\label{optimalNK} 
The optimal values of $N$ and $K$ are found by means of the maximization problem~\eqref{MaximizationZetastar}  in the case of optimal $\zeta^{\star}$. Initially, we consider the integer-relaxed problem where $K$ and $N$ can be any positive scalars, but then, we select the corresponding integer values. 

\begin{theorem}\label{optimalNN} 
 Let the maximization problem~\eqref{Maximization} with  $\lambda_{\mathrm{AP}}$, $K$, and $N$ real variables. For any fixed $\lambda_{\mathrm{AP}},K>0$, the optimal number of AP antennas $N^{\star}$ is given by
 \begin{align}
  N^{\star}=\min \left( \max\left( N_{0},N_{1} \right), N_{2
  } \right)\label{optimalN} 
 \end{align}
with $ N_{0}=\frac{q_{1}-\sqrt{q_{3}}}{q_{2}}$ while $N_{1}=\frac{b_{1}+b_{3}}{b_{4}-b_{2}}$ and $N_{2}=\frac{\frac{K}{\tau_{\mathrm{c}}}b_{1}+b_{3}}{b_{4}-\frac{K}{\tau_{\mathrm{c}}}b_{2}}$,
 where
  $q_{1}=b_{3}^{2}\bar{d}_{11}+b_{1}K \left( b_{1}d_{22}K^{2}+b_{4}\bar{c}_{11}\right)+b_{3} K b_{2}\bar{c}_{11}+b_{1}\left( \bar{d}_{11}+d_{22}K \right) $, $  q_{2}\!=\!2\!\left( b_{3}\!+\!b_{1}K \right)\!\left( b_{4}\!-\!b_{2}K \right)\!\left( b_{2}d_{22}K^{2}\!-\!b_{4}\bar{d}_{11} \right) $, $ q_{3}=q_{1}^{2}+4 q_{2}( b_{2}d_{22}K^{2}-b_{4}\bar{d}_{11} )$ with $\bar{c}_{11}=c_{12}-\left( c_{11}+c_{22} \right)K$,  $c_{11}= \left(\frac{ B_{\mathrm{w} }}{7 L_{\mathrm{AP}}\tau_{\mathrm{c}}}+P_{\mathrm{UE}}\right)/\lambda_{\mathrm{AP}}$, $c_{12}= \left( \rho_{\mathrm{d}}\frac{1-\xi}{\al_{\mathrm{eff}}}-C_{0}\right)/\lambda_{\mathrm{AP}}$, $c_{22}=  \frac{\rho_{\mathrm{tr}}}{\al_{\mathrm{eff}}\lambda_{\mathrm{AP}}}$ ,  ${d}_{22}=\frac{3 B_{\mathrm{w} }  \xi}{ L_{\mathrm{AP}}\lambda_{\mathrm{AP}}}$ $\bar{d}_{11}=\left(D_{0}+D_{1}K\right)/\lambda_{\mathrm{AP}}$, and  the parameters $\{b_{i}\}$  are provided in Table~\ref{ParameterValues3}.
\end{theorem}\begin{proof}
 Similar to the proof of Theorem~\ref{OptimalLambda}, we have  $\zeta^{\star}\tau_{\mathrm{c}}/K=\frac{b_{1}+b_{2}N}{b_{4}N-b_{3}}$, where the parameters $\{b_{i}\}$ are provided in Table~\ref{ParameterValues3}. Then, after substituting this term into the objective function of~\eqref{MaximizationZetastar}, the EE becomes
 \begin{align}
 \!\!\!\mathrm{EE}\!\left(\zeta^{\star} \right)\!=\!\frac{\frac{K \xi}{S} \left( 1-\bar{b}\right)\log_{2}\left( 1+\gamma\right)}{b_{5}\!+\!b_{6} N\!+\!b_{7}\bar{b}\!+\!b_{8}K \xi \left( 1\!- \!\bar{b}\right)\!\log_{2}\!\left( 1\!+\!\gamma\right )},\label{proofN} 
 \end{align}
 which represents a quasi-concave function of $N$. Note that $\bar{b}= \frac{b_{1}+b_{2}N}{b_{4}N-b_{3}} $. The optimal value of $N$ is obtained by computing its first derivative with respect to $N$ and equating it to zero. The resultant value, satisfying the unconstraint problem, is given by~\eqref{optimalN}.
 Taking into account for the constraint in~\eqref{MaximizationZetastar}, this can be written as $ \frac{\tau_{\mathrm{c}}}{K}\le \frac{b_{1}+b_{2}N}{b_{4}N-b_{3}}\le 1$, which results in $N_{1}\le  N^{\star}\le N_{2}$.\end{proof}


\begin{theorem}\label{optimalKK} 
 Let the maximization problem~\eqref{Maximization} with  $ \lambda_{\mathrm{AP}} $, $K$ and $N$ real variables. For any fixed $ \lambda_{\mathrm{AP}} $, $N>0$, the optimal number of users  $K^{\star}$ is given by
 \begin{align}
  K^{\star}=\max\left(K_{2},\max\left(K_{1,1},\min\left(K_{0},K_{1,2}\right)\right)\right),
 \end{align}
 where $ K_{0} $ is one of the real roots of a quintic equation, i.e.,  a polynomial of degree five given by $ \sum_{i=0}^{5} p_{i}x_{i}=0$ with $ p_{0}=e_{1}-e_{2} $, $ p_{1}=\left(e_{1}-e_{2}\right) +2 e_{3}$, $ p_{2}=3\left(e_{3}-2e_{1}\right) $, $ p_{3}=\left(e_{2}+e_{1}\right)-e_{1}e_{3} $, $ p_{4}=2e_{4}\left(e_{1}-e_{2}\right) $, $ p_{5}=\left(e_{2}-e_{1}\right)-e_{1} $. Also, we have $ K_{2}=\frac{e_{4 \frac{K}{\tau_{\mathrm{c}}}}-e_{2}}{e_{1}-e_{3}\frac{K}{\tau_{\mathrm{c}}}} $ and
 \begin{align}
 K_{1,1}&=\frac{-\left(e_{1}-e_{2}\right)-\sqrt{\left(e_{1}-e_{2}\right)^{2}-4 e_{3}e_{4}}}{2 e_{4}}\\
 K_{1,2}&=\frac{-\left(e_{1}-e_{2}\right)+\sqrt{\left(e_{1}-e_{2}\right)^{2}-4 e_{3}e_{4}}}{2 e_{4}}.
 \end{align}
 \end{theorem}
 \begin{table}
\caption{Optimization Parameters for Optimal  Number of Users $K$}
\begin{center}
                                             \begin{tabulary}{\columnwidth}{ | c | c | }\hline
{\bf Parameter} &{\bf Value}\\ \hline\hline
 $e_{1}$&  $\gamma N \rho_{\mathrm{d}} \left( \al-2 \right)$\\ \hline
$e_{2}$ & $\gamma N \rho_{\mathrm{d}} \left(1-3\left(\al-2\right)\right)$ \\ \hline
$e_{3}$&  $2\gamma \left( 1-\al \right)-2 b_{1}$ \\ \hline
 $e_{4}$ &  $-\al \pi N  \rho_{\mathrm{tr}}\tau_{\mathrm{c}} \rho_{\mathrm{d}}\gamma\left( \tau_{\mathrm{tr}}+\lambda_{\mathrm{AP}} -1 \right) $ \\ \hline
  $e_{5}$ & \! $\gamma \rho_{\mathrm{tr}}\tau_{\mathrm{c}}\!\left(\! \left( \al\!-\!2 \right)\!\left(  \tau_{\mathrm{tr}}\!+\!1 \right)\!+\!\al \pi N  \rho_{\mathrm{tr}}\tau_{\mathrm{c}} \rho_{\mathrm{d}}\!\left(\! \left(\!  \tau_{\mathrm{tr}}\!+\!\gamma \right)\! \right)\! \right)\!-\!e_{6}$ \\ \hline
   $e_{6}$ &  $\left( \al-2 \right)\left( \gamma-3  \tau_{\mathrm{tr}} \right) $ \\ \hline
     $e_{7}$ &  $\left(C_{0}+D_{0}N\right)/\lambda_{\mathrm{AP}}$ \\ \hline
     $e_{8}$ &  $\left(P_{\mathrm{COD}} + P_{\mathrm{DEC}} +P_{\mathrm{BT}}\right)/\lambda_{\mathrm{AP}}$ \\ \hline
     $e_{9}$ &  $\left(C_{2}-D_{2}N \right)/\lambda_{\mathrm{AP}}$ \\ \hline
        $e_{10}$ &  $\left(P_{\mathrm{COD}} + P_{\mathrm{DEC}} +P_{\mathrm{BT}}\right)/\lambda_{\mathrm{AP}}$ \\ \hline
 \end{tabulary}\label{ParameterValues4} 
                                         \end{center}
                                     \vskip -5mm
                                         \end{table}
                           \begin{proof}
We notice that the term $A=\sum_{i=1}^{K}|\bpsi_{i}^{\H}\bpsi_{k}|^{2}$, appearing in $\zeta^{\star}\tau_{\mathrm{c}}/K$, depends on $K$ by means of its superscript. In fact, $\zeta^{\star}\tau_{\mathrm{c}}/K$ is an increasing function regarding $A$. Hence, we apply  the  bound on $A$ by using the Welch inequality~\cite{Welch1974}, and we obtain
\begin{align}
 A \ge \frac{\tau_{\mathrm{tr}}\left( K-3 \right)+K-1}{\tau_{\mathrm{tr}}\left( K-2 \right)}\label{welchbound} 
\end{align}
since the summation  becomes $\sum_{i\ne k}^{K}|\bpsi_{i}^{\H}\bpsi_{k}|^{2}=\frac{K-1-\tau_{\mathrm{tr}}}{\tau_{\mathrm{tr}}\left( K-2 \right)}$ by using the inequality. Substituting~\eqref{welchbound} into~\eqref{proofN} and  rearranging with respect to $K$,  the objective funtion  can written as
 \begin{align}
 \!\!\mathrm{EE}\!\left(\zeta^{\star} \right)\!=\!\frac{\frac{K \xi}{S} \left( 1- \bar{e}\right)\log_{2}\left( 1+\gamma\right)}{e_{7}\!+\!e_{8} K\!+\!e_{9} K^{2}\!+\!e_{10} \xi \!\left( 1\!-\! \bar{e}\right)\log_{2}\!\left( 1\!+\!\gamma\right )},\label{EE_K} 
 \end{align}
  where $ \bar{e}=\frac{e_{1}K^{2}+e_{2}K+e_{3}}{e_{4}K^{2}+e_{5}K+e_{6}} $ while the parameters $\{e_{i}\}$ are provided in Table~\ref{ParameterValues4}.
 Taking the  first derivative of~\eqref{EE_K} with respect to $K$ and equating it to zero, we obtain a  polynomial fifth degree with roots provided by an  exhaustive search over the domain set while using a bisection method and the help of Mathematica~\cite{Mathematica}.  We obtain three real roots and one pair of complex roots. Note that the constraint results in $ K_{2} $.\end{proof}

  \section{Numerical Results}\label{Numerical} 
This section presents illustrations of the analytical results provided by means of Theorems~\ref{ReuseFactor}-\ref{optimalKK}  concerning the optimal EE. Notably, the tightness of the derived bounds,  denoting their values as good approximations, is demonstrated in Fig.~\ref{Fig3}  by MC simulations. For the sake of comparison, we have considered a conventional ``cellular'' mMIMO scenario and  a  SCs architecture. Especially,  in SCs, the effective channel power does not harden while in the mMIMO architectures we observe  the signal power tending to its mean as the number of APs becomes large~\cite{Marzetta2010, Ngo2017}. Hence, SCs require both UL and DL training phases, i.e., the length for training in SCs is doubled. Furthermore, CF mMIMO systems enjoy  favorable propagation, and thus, they can achieve optimal performance with simple linear processing. Also, the co-processing, taking place in CF systems, suppresses the inter-cell interference degrading the performance of SCs~\cite{Interdonato2019}.

We consider a sufficiently large squared area of $S=1~\mathrm{km}^{2}$, where the  locations of the   APs  are simulated as realizations of  the PPP $\Phi_{\mathrm{AP}}$ with density $\lambda_{\mathrm{AP}}=100~\mathrm{APs}/\mathrm{km}^{2}$  based on a wraparound topology to keep the translation invariance. We assume that the system bandwidth is $B_{\mathrm{w}}=20~\mathrm{MHz}$ and that each coherence block consists of $\tau_{\mathrm{c}}=200$ samples corresponding to a coherence bandwidth of $ 200~\mathrm{KHz} $ and a coherence  time of $ 1~\mathrm{ms} $ \cite{Ngo2017}.
 Moreover, we assume that   $N=20$ antennas per AP and $K=10$ users in total while $\zeta=4$.  Also, we assume that $ \rho_{\mathrm{tr}}=100~\mathrm{mW}$, $ \rho_{\mathrm{d}}=200~\mathrm{mW}$, $\al=4$, and $\xi=1/3$. Moreover, the normalized UL training transmit power per pilot symbol $\bar{\rho}_{\mathrm{tr}}$ and DL transmit power ${\bar{\rho}_{\mathrm{d}}}$ result by dividing ${\rho}_{\mathrm{tr}}$ and ${{p}_{\mathrm{d}}}$ with the noise power ${N_\mathrm{P}}$ given in $\mathrm{W}$ by ${N_\mathrm{P}} = \kappa_{\mathrm{B}}W_{\mathrm{c}} T_{0}{N_\mathrm{F}}$. For the sake of reference, the descriptions and values of the various system parameters are found in Table~\ref{ParameterValues1} unless otherwise stated. Note that the circuit power parameters have been taken from~\cite{massivemimobook}.
\begin{table*}
\caption{Parameters Values for Numerical Results~}
	\begin{center}
		\begin{tabulary}{\columnwidth}{ | c | c | |c|c|}\hline
			{\bf Description} &{\bf Values}&	{\bf Description} &{\bf Values}\\ \hline
			Number users& $K=10$&Boltzmann constant & $\kappa_{\mathrm{B}}=1.381\times 10^{-23}~\mathrm{J/K}$ \\ \hline
			Number of Antennas/AP & $N=20$&Noise temperature & $T_{0}= 290~\mathrm{K}$ \\ \hline
			AP density & $\lambda_{\mathrm{AP}}=100~\mathrm{APs/km^{2}}$&Noise temperature & $T_{0}= 290~\mathrm{K}$ \\ \hline
						Communication bandwidth & $W_{\mathrm{c}} = 20~\mathrm{MHz}$&Noise figure & ${N_\mathrm{F}}=9~\mathrm{dB}$ \\ \hline	
			Carrier frequency &  $f_{0} = 1.9~\mathrm{GHz}$&Fixed power&  $ P_{\mathrm{FP}}=5~\mathrm{W}$\\ \hline
			Power per pilot symbol & ${\rho}_{\mathrm{tr}}=100~\mathrm{mW}$&Power for AP LO & $P_{\mathrm{LO}}=0.1~\mathrm{W}$\\ \hline
			DL transmit power & ${{\rho}_{\mathrm{d}}}=200~\mathrm{mW}$&Power per AP antenna  &  $P_{\mathrm{AP}}=0.2~\mathrm{W}$\\ \hline
			Path loss exponent & $\al=4$&Power per UE antenna   &  $P_{\mathrm{UE}}=0.1~\mathrm{W}$\\ \hline
			Coherence bandwidth & $B_{\mathrm{c}}=200~\mathrm{KHz}$  
			&Power for data coding   &  $P_{\mathrm{COD}}=0.01~\mathrm{W}/\left( \mathrm{Gbit/\mathrm{s}} \right)$ \\ \hline
						Coherence time & $T_{\mathrm{c}}=1~\mathrm{ms}$
		&Power for data decoding   &  $P_{\mathrm{DEC}}=0.08~\mathrm{W}/\left( \mathrm{Gbit/\mathrm{s}} \right)$
		\\ \hline
		UL training duration& $\tau_{\mathrm{tr}}=10$ samples
		&AP computational efficiency  &   $L_{\mathrm{AP}}=750~\mathrm{Gflops}/\mathrm{W} $\\  \hline
		UL training duration is SCs & $\tau_{\mathrm{tr}}=10$ samples
		&Power for backhaul traffic   &  $P_{\mathrm{BT}}=0.025~\mathrm{W}/\left( \mathrm{Gbit/\mathrm{s}} \right)$\\ \hline
			DL training duration is SCs & $ \tau_{\mathrm{d}}=10$ samples&Power amplifier efficiency  &$  \al_{\mathrm{eff}}=0.5 $ 
			\\ \hline
				\end{tabulary}\label{ParameterValues1} 
	\end{center}
\end{table*}

Firstly, we assess the EE 
by varying the pilot reuse factor $\zeta$ and AP density $\lambda_{\mathrm{AP}} $ for a given pair of $K$, $N$ in a CF mMIMO setting. Specifically, in~Fig.~\ref{Fig1} and in line with Theorem~\ref{ReuseFactor}, it is shown that the EE 
is a pseudo-concave function with respect to $\zeta$ with a unique global maximum at $\zeta^{\star}=3$ while the corresponding optimal EE 
 is $\mathrm{EE}^{\star}=5.92~\mathrm{Mbit/Joule}$. Regarding the AP density, the EE 
  is a quasi-concave function with respect to $\lambda_{\mathrm{AP}}$ as was stated by Theorem~\ref{OptimalLambda}. Nevertheless, this figure shows  the optimal AP density  to achieve maximum EE. Hence, we observe that when $\lambda_{\mathrm{AP}}=25~\mathrm{APs/km^{2}}$, the EE   takes its maximum value. It is worthwhile to mention that the optimal $ \lambda_{\mathrm{AP}} $  depends on fundamental system parameters such as the transmit power and the number of antennas per AP as the corresponding theorem shows.
   
In Fig.~\ref{fig 31}, we illustrate the cellular scenario, where we have assumed an AP with $N=20$ antennas is located per cell and $K=10$ users are served  in total.
In fact, we have relied on a similar work~\cite{Bjoernson2016}, studying the UL transmission of a cellular network with PPP distributed BSs, in order to simulate  the EE  
for the DL. Notably, the outperformance of the CF mMIMO setting is depicted. In particular, in CF mMIMO systems, the EE  
is higher and the required AP density is much lower.  In addition, in the case of SCs, studied in Fig.~\ref{fig 32}, we have  considered the system model in~\cite{Papazafeiropoulos2017}, where  independent users are associated with their nearest multi-antenna AP, while the remaining APs act as interferers. In particular, we have set  $N=4$ antennas per AP serving a single user, i.e., $K=1$. Also,  the imperfect CSI model in that scenario is replaced by the current one while no hardware impairments and channel aging have been assumed. Especially, we have denoted $\bar{\rho}_{\mathrm{tr}}^{\mathrm{sc}}= \bar{\rho}_{\mathrm{tr}}$ and ${\bar{p}_{\mathrm{d}}}^{\mathrm{sc}}=\frac{N}{K}{\bar{p}_{\mathrm{d}}}$, where $\bar{\rho}_{\mathrm{tr}}^{\mathrm{sc}}$ and ${\bar{p}^{\mathrm{d}}}_{\mathrm{sc}}$ are the normalized UL training and DL transmit powers in the case of SCs, in order to guarantee that the total radiated power is  equal in both architectures~\cite{Ngo2017}. Clearly, the  EE   is maximized after a  large AP density, being $ \lambda_{\mathrm{AP}}=70~\mathrm{APs}/\mathrm{km}^{2}$, while in the case of CF systems we need only $ \lambda_{\mathrm{AP}}=25~\mathrm{APs}/\mathrm{km}^{2}$.

 \begin{figure}[!h]
 \begin{center}
 \includegraphics[width=0.9\linewidth]{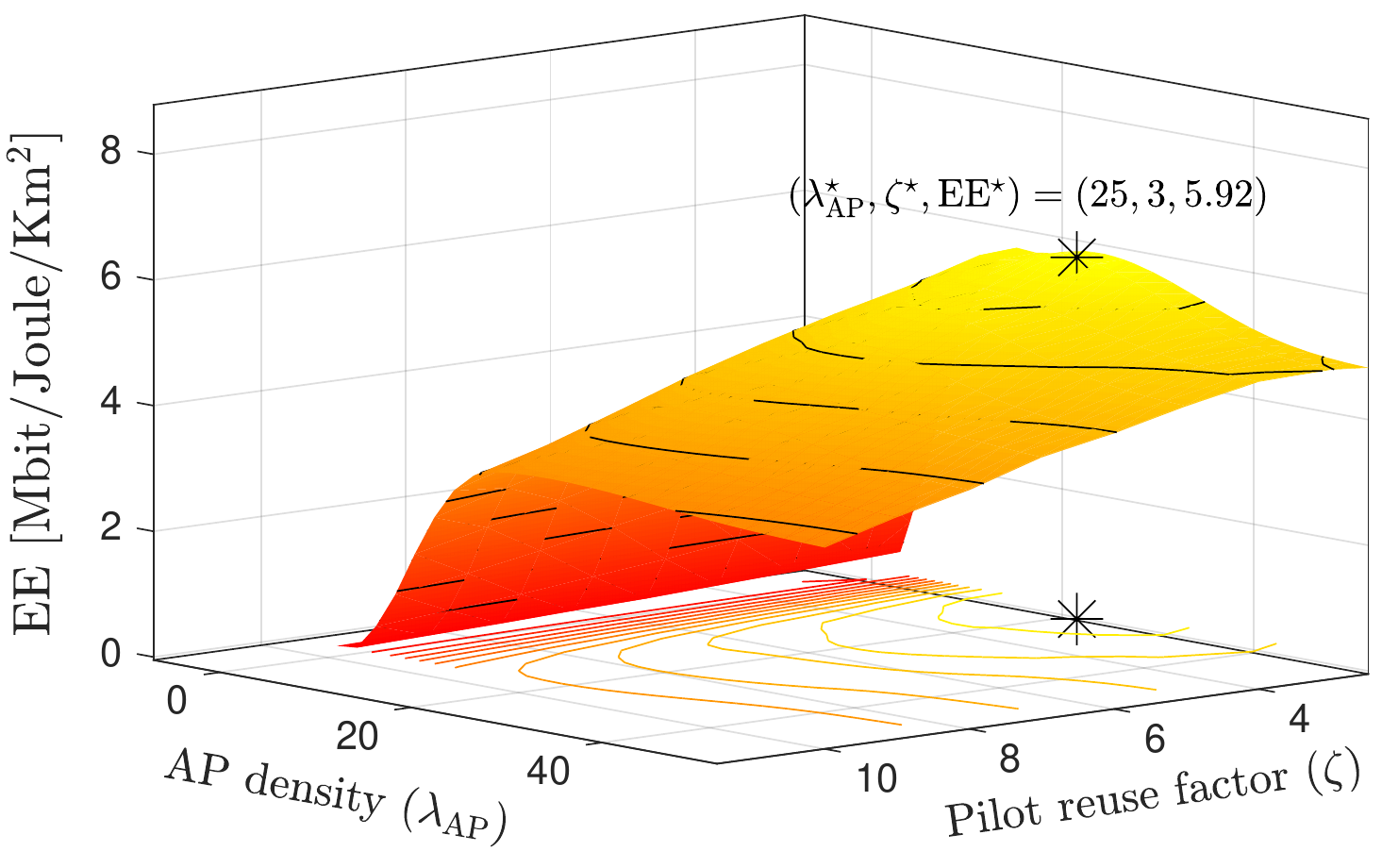}
 \caption{\footnotesize{EE ($\mathrm{Mbit/J}$) of CF mMIMO systems versus the AP density  $\lambda_{\mathrm{AP}}$ and pilot reuse factor $\zeta$. The optimal EE  
 		is star-marked and the corresponding parameters are provided.}}
 \label{Fig1}
 \end{center}
 \end{figure}

 \begin{figure}
 	\centering     
 	\subfigure[``Cellular'' mMIMO systems]{\label{fig 31}\includegraphics[width=0.9\linewidth]{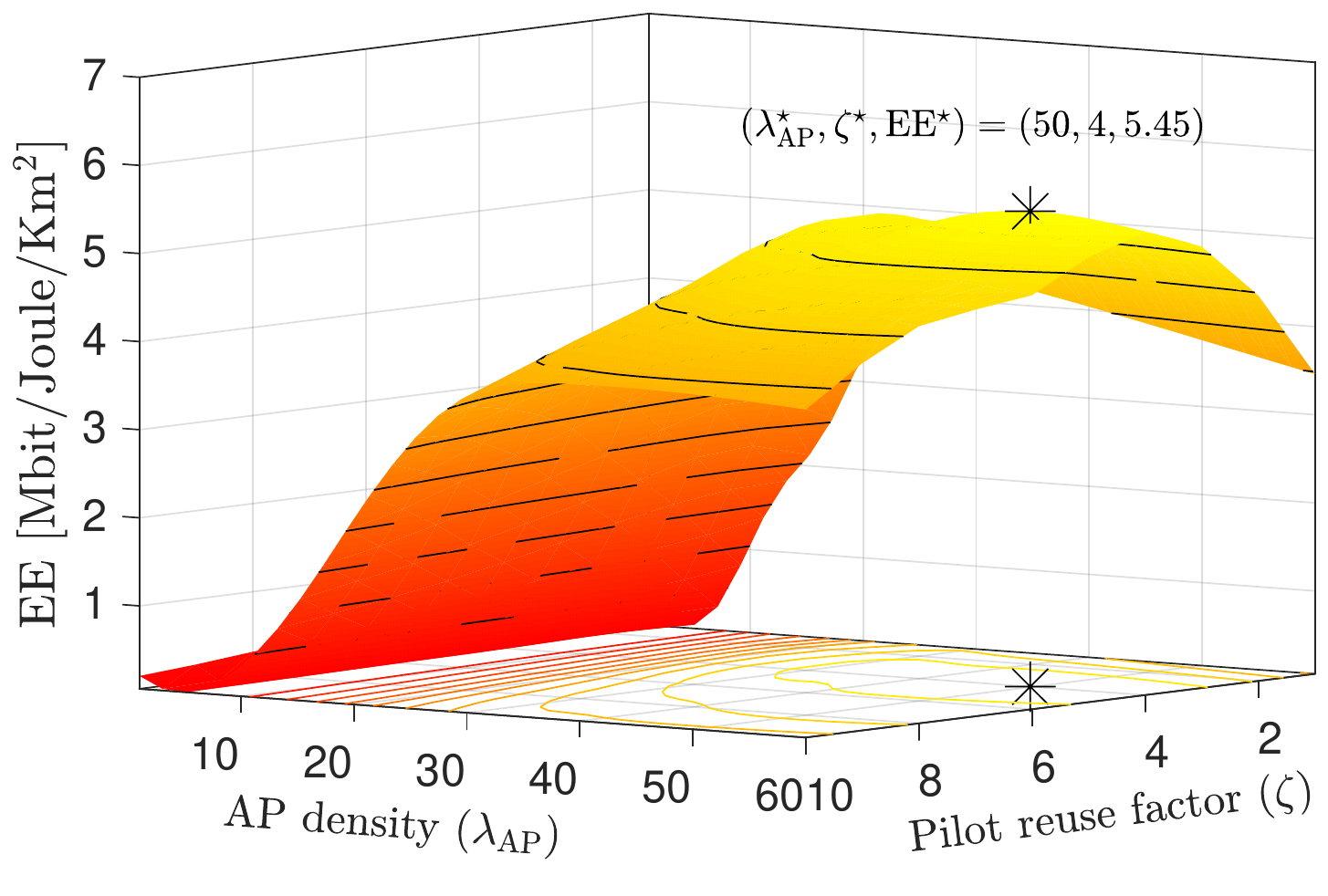}}
 	\subfigure[SCs  systems ]{\label{fig 32}\includegraphics[width=0.9\linewidth]{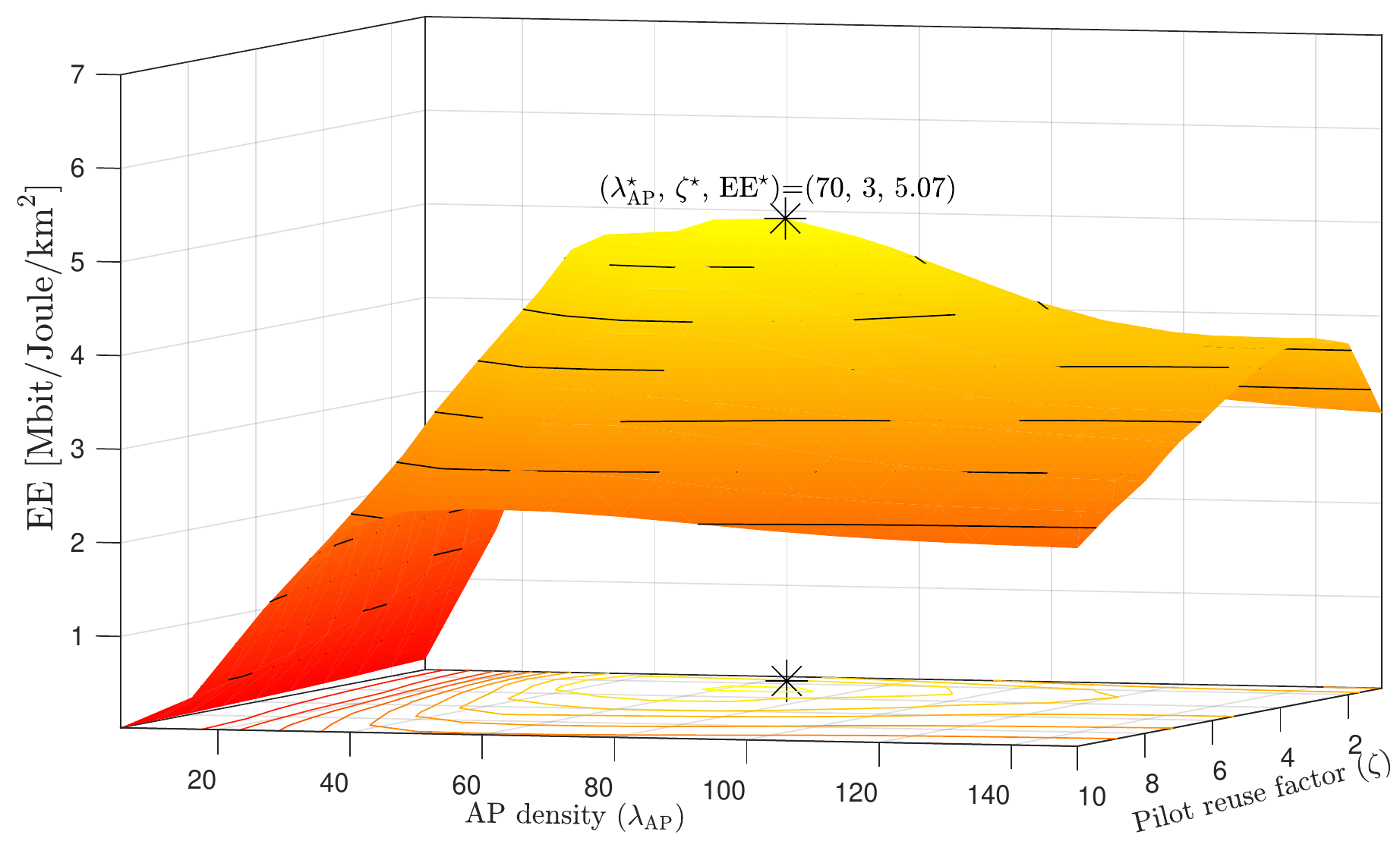}}
 	\caption{EE   ($\mathrm{Mbit/J }$) of  versus the AP density  $\lambda_{\mathrm{AP}}$ and pilot reuse factor $\zeta$ in the cases of a) ``cellular'' mMIMO systems and b) SCs systems, respectively. The optimal EE	is star-marked and the corresponding parameters are provided.}
 \end{figure}

In Fig.~\ref{Fig3}, we examine the impact of the SINR constraint $\gamma_{0}$ on the EE. Moreover, we shed light on the tightness of the lower bound on the average SE given  by Theorem~\ref{PropDetSINRDistances2} and an upper bound provided by averaging the instantaneous SE presented by Lemma~\ref{lemma1} in terms of MC simulations as in \cite{Bjoernson2016}.
 In particular, we assume that $\gamma_{0} \in \{1,~3,~7\}$ to result in an average SE $ \log_{2} \left( 1+ {\gamma}_{0}\right) $ equal to~$1$, $2$, and $3$, respectively. Obviously, the EE  
 decreases with $\gamma_{0}$, which notifies the importance of a target SINR keeping the quality of service in terms of the achievable SE at a satisfactory level according to the specified requirements. Otherwise, we will result in a highly energy-efficient system but useless from the user perspective due to low SE. Furthermore, the gap between the lower and upper bounds is small, which signifies the tightness of the bound proposed by Theorem~\ref{PropDetSINRDistances2}. Hence, the various approximations, employed for the derivation of this bound, provide reliable results. This observation was expected because the bounding techniques provide tight approximations for a large but finite number of APs, which is the case in CF mMIMO systems. Again, it is depicted that the EE 
  increases with $\lambda_{\mathrm{AP}}$ up to a maximum point,  $\lambda_{\mathrm{AP}}=30~\mathrm{APs/km^{2}}$ equivalent to the distance among the APs of $103~\mathrm{m}$ approximately, which is reasonable for practical deployments.

 \begin{figure}[!h]
 \begin{center}
 \includegraphics[width=0.9\linewidth]{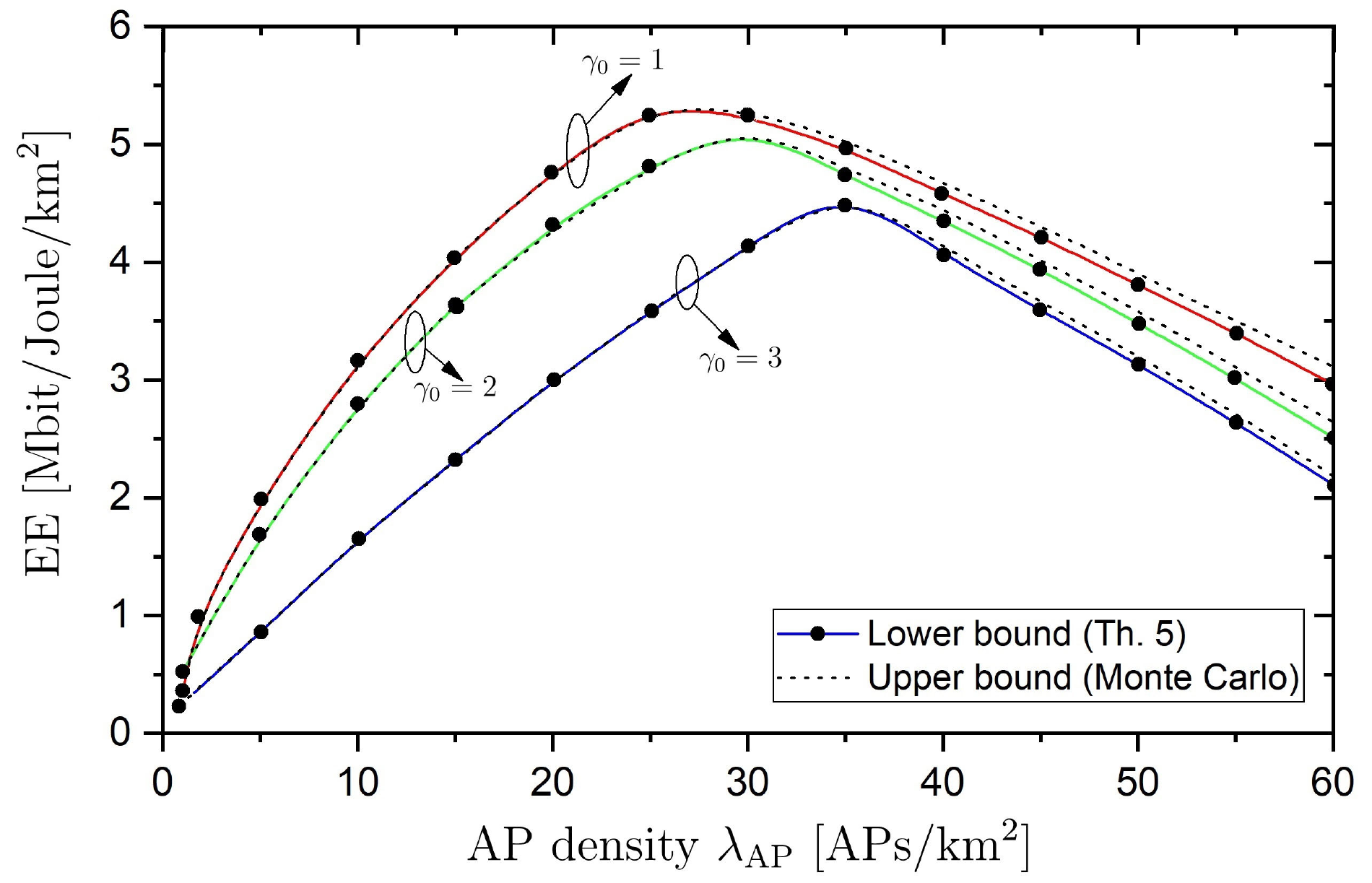}
 \caption{\footnotesize{EE   ($\mathrm{Mbit/J}$) of CF mMIMO systems versus the AP density  $\lambda_{\mathrm{AP}}$ for different SINR constraints. ``Solid-bullet'' and ``dashed'' lines correspond to the lower bound due to the Theorem 1 and upper bound due to MC simulation of the average SE.}}
 \label{Fig3}
 \end{center}
 \end{figure}

In Figs.~\ref{Fig4} and \ref{Fig41}, we depict the EE  
 as a function of the ASE for varying $\gamma_{0}$ for $K=10 $ and $ K=20 $ users, respectively. From both figures, it is obvious that after a certain value of the ASE the EE  
  decreases, while before, the EE  
  increases together with the ASE. Hence,  there are design conditions that could be specified,  in order to achieve maximum EE  
   and ASE simultaneously without sacrificing the one over the other. Also, it is shown that these conditions depend on the SINR constraint $\gamma_{0}$ since its increment, being equivalent to an increase of ASE, results in the decrease of the EE 
    as has been already noticed. Regarding the impact of the number of users, we observe that the EE is reduced with $ K $. This is easily shown by \eqref{Maximization}, which for a given $\gamma_{0}$, $ \mathrm{EE} $is a decreasing function of $ K $ because the ASE increases slower than the increase of the APC with the number of users. However, it is shown that we can achieve the optimal EE at higher ASE because increasing the number of users, the ASE increases.

\begin{figure}
	\centering     
	\subfigure[$ K=10 $ users ]{\label{Fig4}\includegraphics[width=0.9\linewidth]{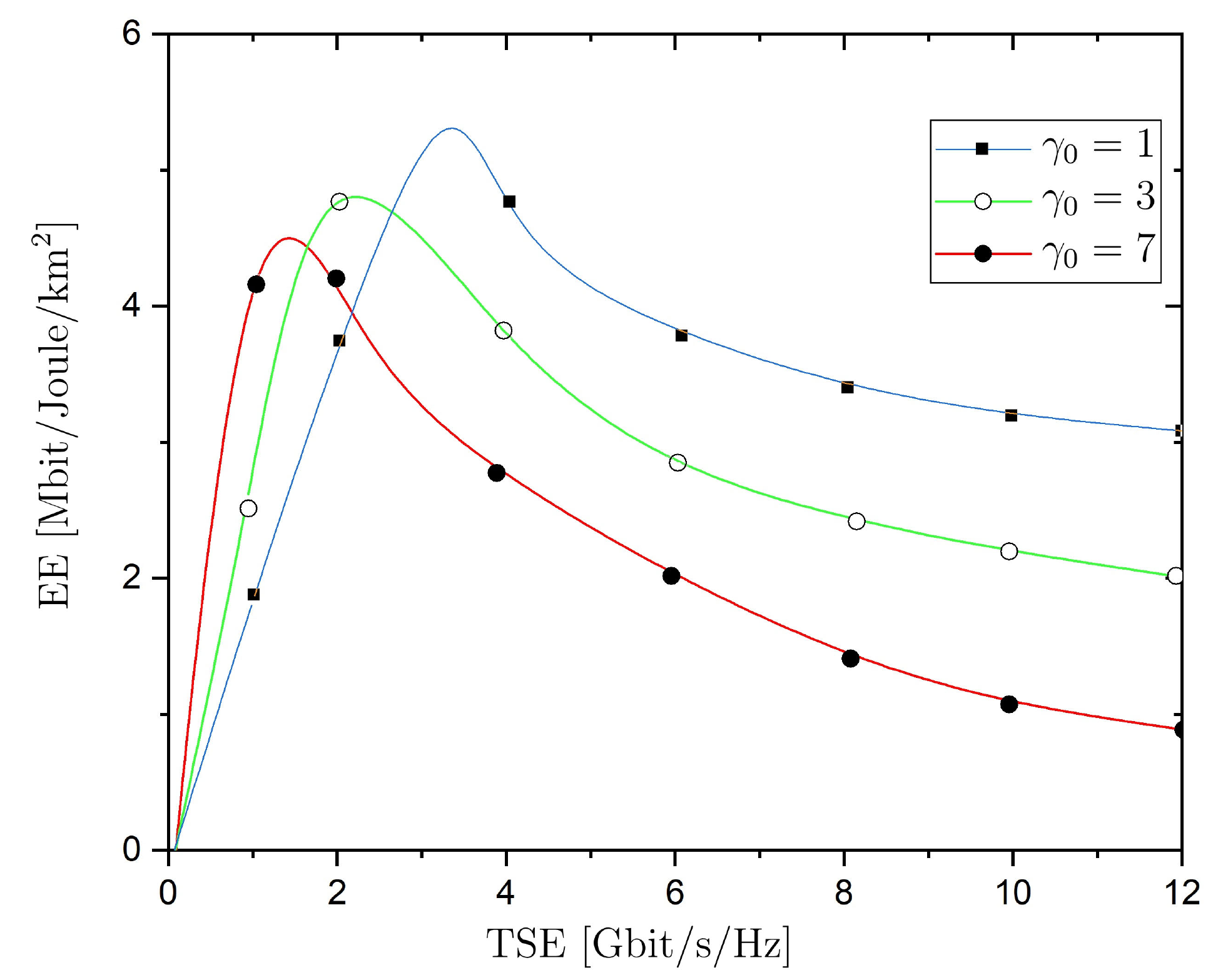}}
	\subfigure[$ K=20 $ users  ]{\label{Fig41}\includegraphics[width=0.9\linewidth]{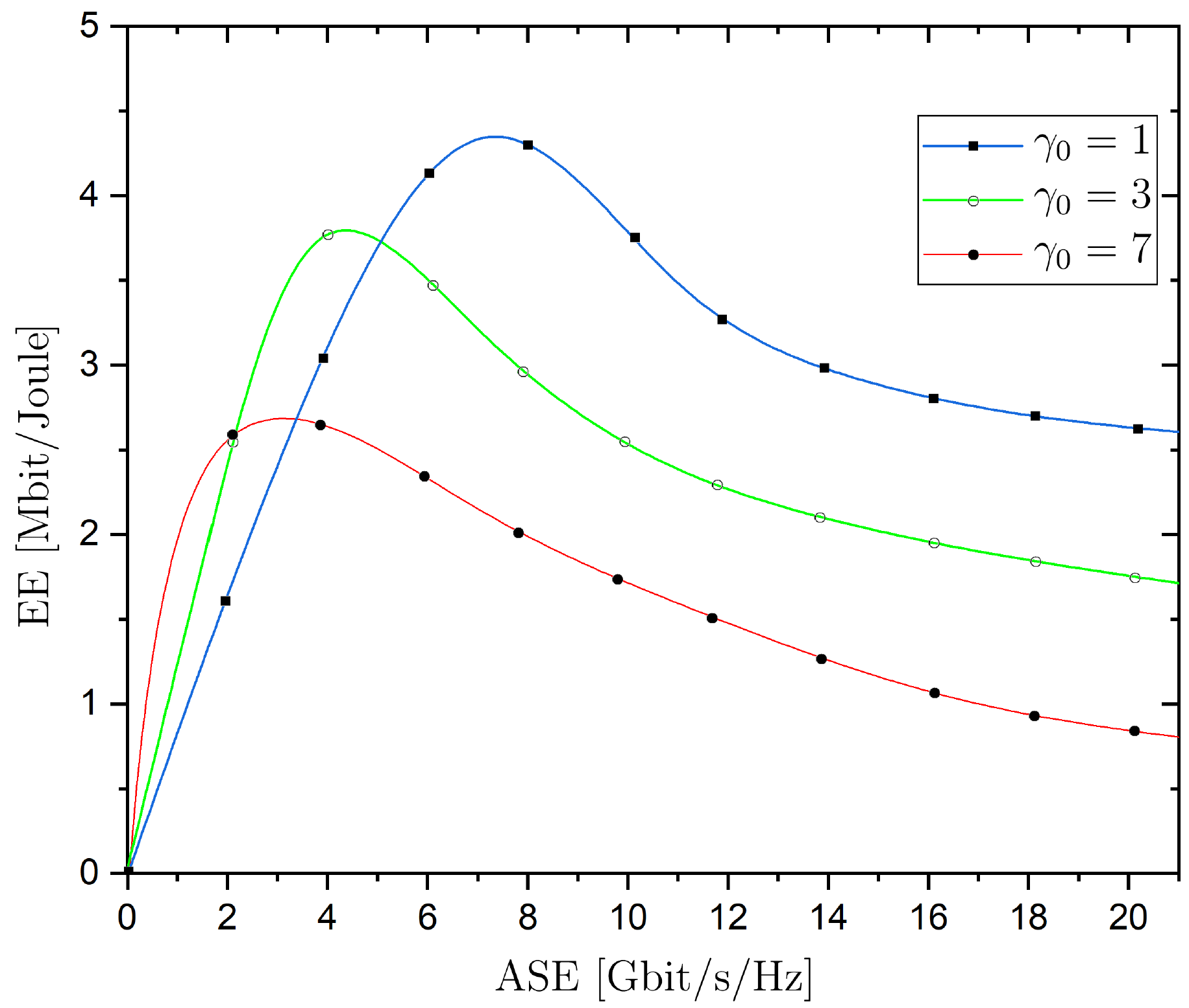}}
	\caption{EE ($\mathrm{Mbit/J}$) of CF mMIMO systems versus the $\mathrm{ASE}$ for different SINR constraints a) $ K=10 $ users and b)  $ K=20 $ users, respectively. }
\end{figure}

Fig.~\ref{Fig5} presents the EE  
 as a function of the  number of antennas per AP $N$ and the number of users $K$ when the average SE is equal to $2~\mathrm{Gbit/s}$, i.e., $\gamma_{0}=3$. Regarding the other parameters under optimization, being the pilot reuse factor and AP density, they are chosen based on Theorems~\ref{ReuseFactor} and Theorem~\ref{OptimalLambda}. Specifically, the optimal values are obtained as $\zeta^{\star}=3$ and  $\lambda_{\mathrm{AP}}^{\star}=25$. We verify that the EE  
  is a pseudo-concave  function of both $K$ and $N$. Based on simulation, the optimal values are given by~$\left( K^{\star},~N^{\star} \right)=\left( 5,~16 \right)$ while the corresponding maximum EE   
  is $\mathrm{EE}^{\star}=6.76~\mathrm{Mbit/Joule}$. Notably, these values are confirmed analytically by means Theorems~\ref{optimalNN} and~\ref{optimalKK}.

  \begin{figure}[!h]
 \begin{center}
 \includegraphics[width=0.9\linewidth]{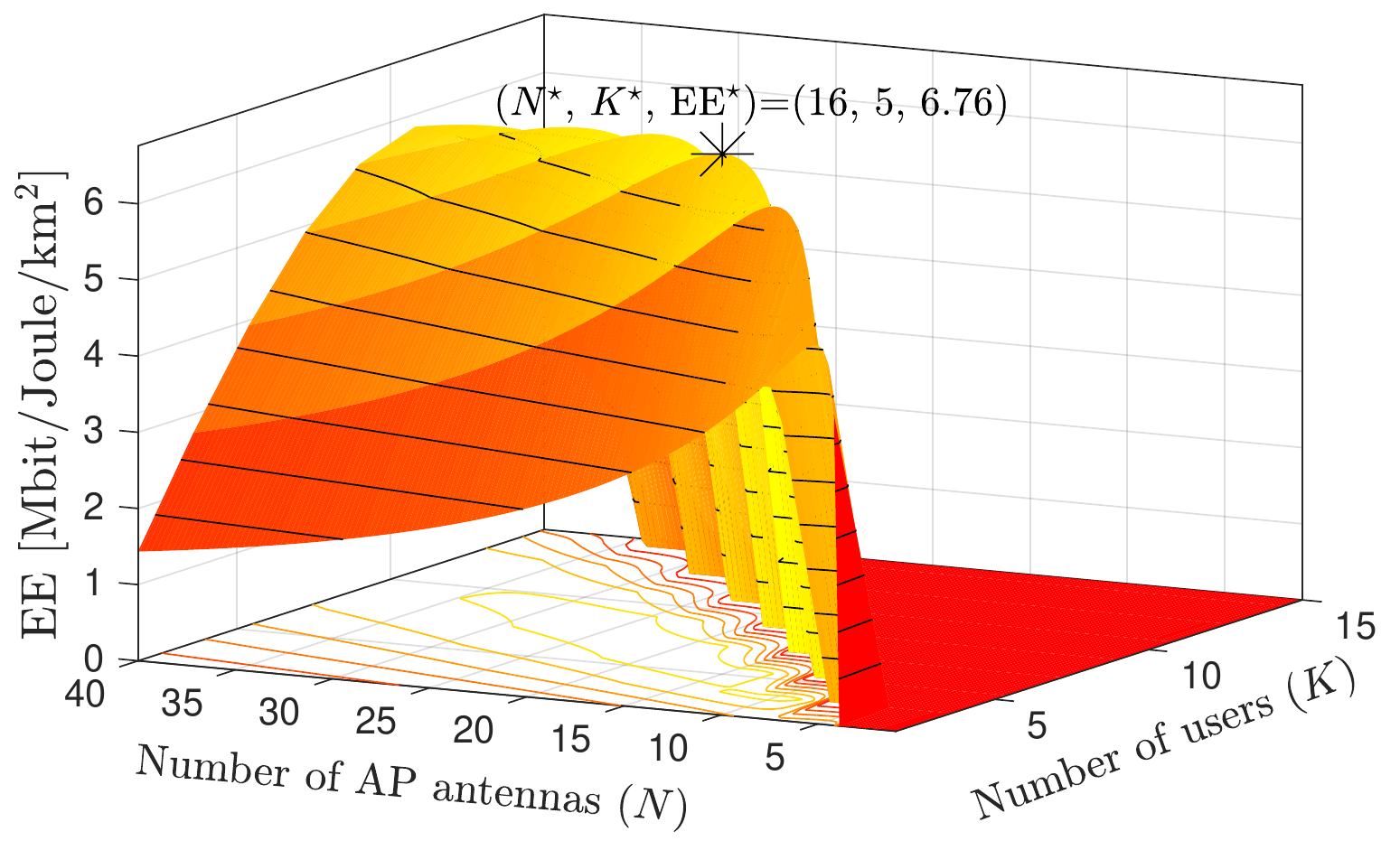}
 \caption{\footnotesize{EE ($\mathrm{Mbit/J}$) of CF mMIMO systems versus the number of AP antennas $N$ and users $K$.  The optimal EE 
 		 has a black triangle and the corresponding parameters are provided.}}
 \label{Fig5}
 \end{center}
 \end{figure}

 \section{Conclusion} \label{Conclusion} 
 Given that network densification is a promising way for high EE, we considered its investigation in a CF mMIMO architecture by assuming both many APs and many antennas per AP. In order to rely on a realistic scenario, we assumed that the APs are PPP distributed. In parallel, we introduced a realistic power consumption model for this setting.  Notably, we achieved to derive a new lower bound on the DL average SE for CF mMIMO systems and  we provided a novel definition for the ASE   which is necessary in the case of CoMP-JT architectures. In this direction, we formulated an EE  maximization problem for  the DL that enabled the analytical determination of tractable closed-form expressions regarding the optimal EE  with respect to the pilot reuse factor, the AP density as well as the number of users and antennas per AP. Remarkably, we achieved to obtain valuable insights concerning the optimization variables. Indeed, the densification  in terms of the AP number up to a specific value increases the EE.
Moreover, the EE is increased up to a certain point by equipping the APs with more antennas due to the higher array and multiplexing gains that manage to mitigate interference and achieve the resultant higher data rate and lower power consumption. 
\begin{appendices}

\section{Proof of Proposition~\ref{PropDetSINR}}\label{SINRproofDistances}
This proof aims at the derivation of $\gamma_{k}$ for finite $M$ by recalling each term of~\eqref{SINR15}. In particular, for the derivation of each term, we are going to use that $\bx^{\H}\by = \tr(\by \bx^{\H})$ for any vectors $\bx$, $\by$.
The term in the numerator becomes 
\begin{align}
 \EE\left[
 \bh_{k}^{\H}\bC_{k}\hat{\bh}_{k}\right]&=\tr\left( \EE\left[\bC_{k}\hat{\bh}_{k} {\bh}_{k}^{\H} \right]  \right) \\
 &=\tr\left( \EE\left[\bL_{k}^{-1}\tilde{\by}_{k} {\bh}_{k}^{\H} \right]  \right)\label{term1}\\
 &=NM\label{term2},
\end{align}
where~\eqref{term1} results after substituting $\hat{\bh}_{k}=\bL_{k} \bD^{-1}\tilde{\by}_{k}$ with 
 $\tilde{\by}_{k}= \big[\tilde{\by}_{1k},\dots,~\tilde{\by}_{Mk}\big]^\T $ from~\eqref{eq:Ypt3} while $\bC_{k}=\bC^{-1}\bD\bL_{k}^{-2}$. The last step is accomplished by applying the expectation between $\tilde{\by}_{k}$ and ${\bh}_{k}^{\H}$. When $i\ne k$, the second-order moment, appearing in the denominator, is written as
\begin{align}
&\EE\left[ \Big| {\bh}_{k}^{\H}\bC_{i}\hat{\bh}_{i}\Big|^{2}\right]=\EE\left[ \Big| \hat{\bh}_{k}^{\H}\bC_{i}\hat{\bh}_{i}\Big|^{2}\right]+\EE\left[ \Big| \tilde{\bee}_{k}^{\H}\bC_{i}\hat{\bh}_{i}\Big|^{2}\right]\label{52} \\
  &=\EE\left[ \Big| \hat{\bh}_{i}^{\H}\bL_{k}\bL_{i}^{-1}\bC_{i}\hat{\bh}_{i}\Big|^{2}\right]+\EE\left[ \Big| \tilde{\bee}_{k}^{\H}\bC_{i}\hat{\bh}_{i}\Big|^{2}\right]\label{53} \\
   &=N^{2}\!\left( \tr ^{2}\!\!\left( \bL_{k}\bL_{i}^{-1}\right)+\tr\!\left(  \bL_{k}^{2}\bL_{i}^{-2}  \right)+\tr\!\left(\bC_{i}\! \left( \bL_{k}-\bD^{-1}\bL_{k}^{2} \right) \right) \right)\nn\\
   &=N^{2}\left( \tr ^{2}\left( \bL_{k}\bL_{i}^{-1}\right)+\tr\left(\bC_{i} \bL_{k}\right) \right)\!,
\end{align}
where in~\eqref{52}, we have used that $\bh_{k}=\hat{\bh}_{k}+\tilde{\bee}_{k}$ and the identity $\EE\left[ |X+Y|^{2}\right] =\EE\left[ |X|^{2}\right] +\EE\left[ |Y^{2}|\right]$ holding between two independent random variables with $\EE\left[ X\right]=0 $.  In~\eqref{53}, we have applied the property concerning the  estimated channels between pilot contaminated users, i.e., $\hat{\bh}_{k}=\bL_{k}\bL^{-1}\hat{\bh}_{i}$~\cite{Nayebi2016}. The first part of the next equality follows by using \cite[Lemma~2]{Bjornson2015}, and the second part, not depending on the contamination,  results due to the independence between the two random vectors. The last equation is obtained by simple algebraic manipulations.
On the contrary, if $i=k$, we have
\begin{align}
  \EE\left[ \Big| {\bh}_{k}^{\H}\bC_{i}\hat{\bh}_{i}\Big|^{2}\right]&=\EE\left[ \Big| \hat{\bh}_{k}^{\H}\bC_{k}\hat{\bh}_{k}\Big|^{2}\right] \\
   &=\tr ^{2}\Id_{MN}+\tr\Id_{MN}\label{55}\\
   &=M^{2}N^{2}+MN,
\end{align}
where in~\eqref{55}, we have applied \cite[Lemma~2]{Bjornson2015}.
%
In total, we have
\begin{align}
 \EE\left[ \Big| {\bh}_{k}^{\H}\bC_{i}\hat{\bh}_{i}\Big|^{2}\right]=N^{2}\tr^{2}\left(\bL_{i}^{-1}  \bL_{k}\right)
 +\begin{cases}
                                     N^{2}\tr \left(\bC_{i}\bL_{k}\right), &             i\ne k\\
 MN, &		i=k.
                                      \end{cases}\label{term5}
\end{align}
Also, the normalization parameter can be easily written as
\begin{align}
\mu &= \frac{K}{\EE\Big[\sum_{i=1}^{K}\hat{\bh}_{i}^{\H}\bC_{i}^{2}\hat{\bh}_{i}\Big]}\nn\\
&= \left( \frac{N}{K}\sum_{i=1}^{K}\tr\bC_{i} \right)^{-1}\label{term6}
\end{align}
The proof is concluded by susbstituting~\eqref{term2},~\eqref{term6} and~\eqref{term5} into~\eqref{SINR15}.

\section{Proof of Theorem~\ref{PropDetSINRDistances2}}\label{SINRproofDistances2}
The proof starts with the application of Jensen's inequality that will allow us to derive a tractable lower bound of the  DL average SE by moving the expectation inside the logarithm and continues with  the derivation of the expectation of the inverse SINR over the APs dinstances.\\
Application of the Jensen inequality to the DL average SE results in 
\begin{align}
 \EE\left[\log_{2}\left( 1+\frac{1}{\bar{\gamma}_{k}^{-1}} \right) \right] \ge \log_{2} \left( 1+\frac{1}{\check{\gamma}_{k}} \right)\!,
\end{align}
where the expectation
applies directly to the inverse SINR  $\check{\gamma}_{k}={\EE\left[\bar{\gamma}_{k}^{-1} \right]}$.\\
The inverse  SINR provided by~\eqref{SINRFinite} can be written as \eqref{SINRFiniteproof11} at the top of the next page.
	\begin{figure*}
\begin{align}
 \EE\left[ \bar{\gamma}_{k}^{-1} \right] &\!=\!\EE\Bigg[\frac{\displaystyle\sum_{i=1}^{K}\sum_{m=1}^{M} \!\!\left(d_{m}l_{mi}^{-2}\!\left(\!N l_{mk}\!+\!\frac{1}{K p_{\mathrm{d}}}\right)\!\!\right)\!+\!N\sum_{i \ne k}^{K}\!\!\left( \sum_{m=1}^{M}\! l_{mk}l_{mi}^{-1} \right)^{\!2}\!-\!N\sum_{m=1}^{M} \!d_{m}l_{mk}^{-1}\!+\!M}{M^{2}N}\Bigg] \!,\label{SINRFiniteproof11}
\end{align}
	\hrulefill
\end{figure*}
where   the trace of each matrix is replaced by the sum of its entry-wise  elements.
For the derivation of  the expectation, let a ball of radius $R$ centered at the  origin that contains $M=\Phi\left( B\!\left( o,R \right) \right)$ points with $ S=|B\!\left( o,R \right)\!|$. The first step  includes conditioning on this area of radius $R$ and on the number of points in this  area. Next, we apply  the law of large numbers.
Afterwards, we remove the conditioning regarding the number of points, and  we assume that  the area is  infinite, i.e., $R \to \infty$. Specifically, we have
\begin{align}
 \!\!\!\EE\!\left[ \bar{\gamma}_{k}^{-1}\right]
 &\!=\!\!\lim_{R \to \infty} \!\! \EE\Bigg[\!\frac{1}{{M^{2}N}}{\displaystyle\sum_{i=1}^{K}\!\sum_{m\in \Phi_{\mathrm{AP}}\cap B\left( o,R \right)}^{M}\! \!\!\!\!\!\!\!\!\!\!d_{m}l_{mi}^{-2}\!\left(\!\!N l_{mk}\!+\!\frac{1}{K p_{\mathrm{d}}}\right)}{}\!\!\Bigg]\!\!\nn\\
 &+\!\lim_{R \to \infty} \! \EE\Bigg[\frac{1}{{M^{2}}}{\!\displaystyle\sum_{i \ne k}^{K}\!\!\left( \sum_{m\in \Phi_{\mathrm{AP}}\cap B\left( o,R \right)}^{M}\!\!\!\!\!\!\!l_{mk}l_{mi}^{-1}\! \right)^{\!\!\!\!2}}{}\Bigg]\nn\\
 &  +\lim_{R \to \infty} \! \EE\Bigg[\frac{1}{M^{2}}\!\!\!\sum_{m\in \Phi_{\mathrm{AP}}\cap B\left( o,R \right)}^{M} \!\!\!\!\!\!\!\!d_{m}l_{mk}^{-1}\!\Bigg]\!+\!\EE\left[ \frac{1}{MN}\right], \!\label{SINRFiniteproof13}
  \end{align}
where in \eqref{SINRFiniteproof13}, we have let the ball of radius $R$ going to infinity. We continue with the computation of the first term of~\eqref{SINRFiniteproof13}. We have
\begin{align}
& \mathcal{I}_{1}=\lim_{R \to \infty} \! \EE\Bigg[\frac{1}{{M^{2}N}}{V}{}\Bigg]\nn\\
 &=\lim_{R \to \infty} \EE_{M}\!\left[\!\EE_{|M} \!\!\left[\frac{1}{{M^{2}N}}{V}{}|M=\Phi\!\left( B\left( o,R \right) \right)\right]\right] \label{SINRFiniteproof14}\\
  &=   \sum_{i=1}^{K}\lim_{R \to \infty} \EE_{M}\!\left[\!\EE\!\left[  M\frac{1}{{MN}}{d_{m}l_{mi}^{-2}\!\left(\!\!N l_{mk}\!+\!\frac{1}{K p_{\mathrm{d}}}\right)}{}\right] \right]  \label{SINRFiniteproof15}\\
  &=   \sum_{i=1}^{K}\EE\!\left[  \frac{1}{{N}}{d_{m}l_{mi}^{-2}\!\left(\!\!N l_{mk}\!+\!\frac{1}{K p_{\mathrm{d}}}\right)}{}\right] \label{SINRFiniteproof16}\\
  &=\mathcal{I}_{11}+\mathcal{I}_{12},\label{SINRFiniteproof161}
\end{align}
where  $ V\!=\!\displaystyle\sum_{i=1}^{K}\!\sum_{m\in \Phi_{\mathrm{AP}}\cap B\left( o,R \right)}^{M}\! \!\!\!\!\!\!\!\!\!\!\!\!d_{m}l_{mi}^{-2}\!\left(\!\!N l_{mk}\!+\!\frac{1}{K p_{\mathrm{d}}}\right) $, $ \mathcal{I}_{11}\!=\!  \displaystyle\!\sum_{i=1}^{K}\!\EE\!\left[  d_{m}l_{mi}^{-2} l_{mk}\right]$, and 
$ \mathcal{I}_{12}\!= \! \frac{1}{{KN p_{\mathrm{d}}}}\!\displaystyle \sum_{i=1}^{K}\!\EE\!\left[ {d_{m}l_{mi}^{-2}}{}\right]$.
In \eqref{SINRFiniteproof14}, we compute the conditional expectation given the number of points inside the ball, while in \eqref{SINRFiniteproof15}, one $M$ in the denominator cancels out with the number of points inside the ball, and then, we  apply the law of large numbers given the number of APs.  Next, we derive $\mathcal{I}_{11}$ as
\begin{align}
 &\!\!\mathcal{I}_{11}= \EE\left[ \sum_{i=1}^{K}\left( \sum_{j=1}^{K}|\psi_{j}\psi_{k}^{\H}|^{2}l_{mj}+\frac{1}{{\tau_{\mathrm{tr}} \rho_{\mathrm{tr}}}} \right)l_{mi}^{-2}l_{mk}\right]\label{SINRFiniteproof162}\\
 &\!\!=\! \EE\!\!\left[\!\! \sum_{i=1}^{K}\!\! \sum_{j=1}^{K}\!\!|\psi_{j}\psi_{k}^{\H}|^{2} l_{mj}l_{mi}^{-2}l_{mk}\!\right]\!+\!\frac{1}{{\tau_{\mathrm{tr}} \rho_{\mathrm{tr}}}}\EE\!\left[  \sum_{i=1}^{K} l_{mi}^{-2}l_{mk}\right]\label{Expectationequation}\!, 
\end{align}
where the first part of~\eqref{Expectationequation}  can be written as
\begin{align}
 &\EE\left[ \sum_{i=1}^{K} \sum_{j=1}^{K}|\psi_{j}\psi_{k}^{\H}|^{2} l_{mj}l_{mi}^{-2}l_{mk}\right]
\nn\\
&=\left\{\begin{array}{ll}
 \sum_{i=1}^{K}  |\psi_{i}\psi_{k}^{\H}|^{2}    \EE \left[l_{mi}^{-1}l_{mk}\right]&~\mathrm{if}~j=i\\
 \sum_{i=1}^{K}\sum_{k=1}^{K}  \EE \left[  l_{mi}^{-2}l_{mk}^{2} \right] &~\mathrm{if}~j= k\\
 \sum_{j\ne i,k}^{K}
 |\psi_{j}\psi_{k}^{\H}|^{2}  \EE \left[ l_{mj}  l_{mi}^{-2}l_{mk} \right] &~\mathrm{otherwise}
 \end{array}.\!\label{Expectationequation1}
\right.
\end{align}
If $i\ne k$, the expectation in the first branch  results in
\begin{align}
 \EE \left[l_{mi}^{-1}l_{mk} \right]&=  \EE \left[ \frac{1}{ l_{mi}^{}}\right] \EE\left[ l_{mk}^{}\right] \label{firstPart1} \\
 & \ge \frac{1}{\EE \left[  l_{mi}^{}\right]}\EE\left[  l_{mk}^{}\right]\label{firstPart3}\\
 &=1\label{firstPart2},
\end{align}
where~\eqref{firstPart1} is obtained due to the independence between the random variables $l_{mi}$ and $l_{mk}$ while~\eqref{firstPart3} has accounted for Jensen's inequality. Notably,~\eqref{firstPart2} is obtained since the two variables have the same marginal distribution. In the condition that $i=k$, the result is the same. Following the same procedure, the expectation in the second branch gives the same result. The expectation in the  last branch becomes
\begin{align}
\EE \left[  l_{mj}^{}  l_{mi}^{-2}l_{mk}^{} \right] &=   \left\{\begin{array}{ll}
                  \EE \left[  l_{mj}^{}  l_{mk}^{-1} \right]&~\mathrm{if}~i=k\\
                   \EE \left[  l_{mj}^{}  l_{mi}^{-2}l_{mk}^{} \right]&~\mathrm{if}~i\ne k
                   \end{array}\label{Expectationequation2}
\right. .
            \end{align}
The first branch is identical to~\eqref{firstPart1}, and gives the same expression. The second branch gives
\begin{align}
\EE \left[  l_{mj}^{}  l_{mi}^{-2}l_{mk}^{} \right]&=\EE \left[  l_{mj}^{} \right] \EE\left[  l_{mi}^{-2}\right] \EE\left[ l_{mk}^{} \right]\label{thirdPart1}\\
&\ge \EE \left[  l_{mj}^{} \right] \EE\left[  l_{mi}^{-1}\right] ^{2}\EE\left[ l_{mk}^{} \right]\label{thirdPart2}\\
&\ge 1,\label{thirdPart3}
\end{align}
In~\eqref{thirdPart1}, we have taken into consideration the independence among the variables, and then, in~\eqref{thirdPart2} we have applied  the inequality $\EE\left[ x^{2}\right] \ge \EE\left[ x\right]^{2} $. Eq.~\eqref{thirdPart3} is obtained after following similar steps with~\eqref{firstPart2}.
The second part of~\eqref{Expectationequation} is written as
\begin{align}
 \EE\left[ \sum_{i=1}^{K}\ l_{mi}^{-2}l_{mk}\right]&=  \left\{\begin{array}{ll}
                           \EE\left[ l_{mi}^{-1}\right]&~\mathrm{if}~i=k  \\
                           \sum_{i\ne k}^{K}\EE\left[  l_{mi}^{-2}l_{mk}^{}\right]&~\mathrm{if}~i\ne k \end{array}.\label{Expectationequation3}
\right. \end{align}
Now, the first branch for a general power $q$ becomes
\begin{align}
 \EE\left[  l_{mi}^{-q}\right]
 &\ge  \frac{1}{\EE\left[l_{mi}^{q}\right]},
\end{align}
where we have applied Jensen's inequality. Note that
\begin{align}
 \EE\left[l_{mi}^{q}\right]&=2\pi \left( \int_{0}^{1}{ y} \mathrm{d}y+\int_{1}^{ \infty} y^{-qa+1} \mathrm{d}y \right) \label{eq51} \\
 &=\frac{ q \al \pi}{ q \al-2 }. \label{eq6}
\end{align}
Regarding the second branch in~\eqref{Expectationequation3}, we have
\begin{align}
 \EE\left[  l_{mi}^{-2}l_{mk}^{}\right]&=\EE\left[  l_{mi}^{-2}\right] \EE\left[l_{mk}^{}\right]\\
 &\ge \EE\left[  l_{mi}^{-1}\right]^{2} \EE\left[ l_{mk}^{}\right]\label{ineq1} \\
 &\ge \frac{\EE\left[ l_{mk}^{}\right]}{\EE\left[  l_{mi}^{}\right]^{2}}\label{ineq2} \\
 &=\frac{1}{\EE\left[  l_{mi}^{}\right]}\\
 &=\frac{ \al-2 }{ \al\pi}\label{ineq3},
\end{align}
where we have applied a property of variance in~\eqref{ineq1}, and the Jensen's inequality in~\eqref{ineq2}. Next, in~\eqref{ineq3}, we have used~\eqref{eq6}. 
%
With respect to the second part of~\eqref{SINRFiniteproof161} and by following a similar procedure, we have
\begin{align}
 \!\!\mathcal{I}_{12}&= \frac{1}{{KN p_{\mathrm{d}}}}\EE\!\left[ \sum_{i=1}^{K}\!\left( \sum_{j=1}^{K}|\psi_{j}\psi_{k}^{\H}|^{2}l_{mj}+\frac{1}{{\tau_{\mathrm{tr}} \rho_{\mathrm{tr}}}} \!\right)\!l_{mi}^{-2}\right]\label{SINRFiniteproof165}\\
 &= \frac{1 }{\al\pi N p_{\mathrm{d}} }\left( \sum_{j=1}^{K}|\psi_{j}\psi_{k}^{\H}|^{2}{\left( \al-2 \right)}{ }+\frac{\al-1}{{\tau_{\mathrm{tr}} \rho_{\mathrm{tr}}}} \right)\!.\label{SINRFiniteproof166}
\end{align}
Substituting the results concerning  $\mathcal{I}_{11}$ and $\mathcal{I}_{12}$, we obtain $\mathcal{I}_{1}$. The second term in~\eqref{SINRFiniteproof13} becomes
\begin{align}
 &\!\!\!\!\mathcal{I}_{2}=\lim_{R \to \infty} \! \EE\Bigg[\frac{1}{{M^{2}}}{\displaystyle\sum_{i \ne k}^{K}U_{i}}{}\Bigg]\nn\\
 &\!\!\!\!=\lim_{R \to \infty} \EE_{M}\!\left[\!\EE_{|M} \!\!\left[\frac{1}{{M^{2}}}{\displaystyle\sum_{i \ne k}^{K}\!\!U_{i}}|M=\Phi\!\left( B\left( o,R \right)\right)\right]\right] \label{SINRFiniteproof17}\\
  &\!\!\!\!=\!\lim_{R \to \infty}\!\! \EE_{M}\!\!\left[\!\EE_{|M} \!\!\left[{\displaystyle\sum_{i \ne k}^{K}\!\!\left( \EE\left[ l_{mk}l_{mi}^{-1}\right]\right)^{\!2}}|M=\Phi\!\left( B\left( o,R \right)\right)\!\right]\!\right] \label{SINRFiniteproof18}\\
  &=\lim_{R \to \infty} \frac{1}{|B\left( o,R \right)|}\EE_{M}\left[ M\right] \displaystyle\sum_{i \ne k}^{K}\left( \EE\left[ l_{mk}l_{mi}^{-1}\right]\right)^{\!2} \label{SINRFiniteproof18}\\
  &= \lambda_{\mathrm{AP}}\displaystyle\sum_{i \ne k}^{K}\left( \EE\left[ l_{mk}l_{mi}^{-1}\right]\right)^{\!2} ,\label{SINRFiniteproof19}\\
  &= \lambda_{\mathrm{AP}}\left( K-1 \right)\!,\label{SINRFiniteproof191}
\end{align}
where $ U_{i}\displaystyle=\left( \sum_{m\in \Phi_{\mathrm{AP}}\cap B\left( o,R \right)}^{M}\!\!\!l_{mk}l_{mi}^{-1}\! \right)^{\!\!\!\!2} $. In \eqref{SINRFiniteproof17}, we have applied the law of large numbers,  and in~\eqref{SINRFiniteproof19} we have taken into account that  $\EE_{M}\left[ M\right]= \lambda_{\mathrm{AP}} |B\left( o,R \right)\!|$. In~\eqref{SINRFiniteproof191}, we have used similar steps to~\eqref{firstPart2}. Similarly, the third term in~\eqref{SINRFiniteproof13}  is obtained as 
\begin{align}
  \mathcal{I}_{3}&=\lim_{R \to \infty} \! \EE\Bigg[\frac{1}{M^{2}}\!\!\!\sum_{m\in \Phi_{\mathrm{AP}}\cap B\left( o,R \right)}^{M} \!\!\!\!\!\!\!\!d_{m}l_{mk}^{-1}\!\Bigg]\label{SINRFiniteproof20}\\
  &=\EE\left[ \left( \sum_{j=1}^{K}|\psi_{j}\psi_{k}^{\H}|^{2}l_{mj}+\frac{1}{{\tau_{\mathrm{tr}} \rho_{\mathrm{tr}}}} \right)l_{mk}^{-1}\right]\label{SINRFiniteproof167}\\
 &= \sum_{j=1}^{K}|\psi_{j}\psi_{k}^{\H}|^{2}+\frac{\al-2}{{\al \pi\tau_{\mathrm{tr}} \rho_{\mathrm{tr}}}},\label{SINRFiniteproof168}
\end{align}
where in \eqref{SINRFiniteproof167}, we have followed similar steps to \eqref{SINRFiniteproof15}.
Regarding the last term in~\eqref{SINRFiniteproof13}, we have 
\begin{align}
  \mathcal{I}_{4}&=\lim_{R \to \infty}\EE\left[ \frac{1}{MN}\right]\\
  &\ge \lim_{R \to \infty}\frac{1}{N\EE\left[M \right] }\label{SINRFiniteproof21}\\
   &=\lim_{R \to \infty}\frac{1}{N\lambda_{\mathrm{AP}} |B\left( o,R \right)|  }\label{SINRFiniteproof22}\\
   &=0,
\end{align}
where in~\eqref{SINRFiniteproof21} we have applied Jensen's inequality. Next, we have used that  $\EE_{M}\left[ M\right]= \lambda_{\mathrm{AP}} |B\left( o,R \right)\!|$, and we have computed the limit $R\to \infty$. Substituting $\mathcal{I}_{1}$, $\mathcal{I}_{2}$, $\mathcal{I}_{3}$, and $\mathcal{I}_{4}$ into~\eqref{SINRFiniteproof13}, we conclude the proof.

\section{Proof of Proposition~\ref{PropAPC}}\label{AreaPowerConsumptionLemmaproof}
The proof, split in two parts, starts with  the expression of $P_{\mathrm{TX}}$ by means of a lemma, and continues with the presentation of $P_{\mathrm{CPC}}$. 
\begin{lemma}\label{AreaPowerConsumptionLemma}
The total average  transmit power consumption due to UL pilot and DL data transmissions of an arbitrary AP is 
\begin{align}
P_{\mathrm{TX}}= K\frac{K/\zeta \rho_{\mathrm{tr}}+\tau_{\mathrm{d}} \rho_{\mathrm{d}}}{\tau_{\mathrm{c}}},\label{Ptx1} 
\end{align}
where $\tau_{\mathrm{d}}=\xi\left(\tau_\mathrm{c}-  \tau_{\mathrm{tr}}\right)$. 
\end{lemma}
\begin{proof}
In each coherence block, each user transmits pilot symbols for a fraction of $ \tau_{\mathrm{tr}}/ \tau_{\mathrm{c}}$ with power $ \rho_{\mathrm{tr}}$, 
 while each AP trasmits data symbols for a fraction of $ \tau_{\mathrm{d}}/ \tau_{\mathrm{c}}$ with power $\rho_{\mathrm{d}}$.
\end{proof}
The second part of~\eqref{APC1}, concerning the  $P_{\mathrm{CPC}}$ of an arbitrary AP, is given by~\cite{Pizzo2018}
\begin{align}
P_{\mathrm{CPC}}=P_{\mathrm{FP}}+P_{\mathrm{TC}}+P_{\mathrm{C-BC}}+P_{\mathrm{CE}}+P_{\mathrm{LP}},\label{Pcpc} 
\end{align}
where these terms correspond to the power consumptions of circuitry parts. Specifically,  $P_{\mathrm{FP}}$ expresses the power consumed for site-cooling and  control signaling  and the traffic-independent mixed power consumption of each backhaul, $P_{\mathrm{TC}}$ for the transceiver chain, $P_{\mathrm{C-BC}}$ for coding and load-dependent backhauling cost, while $P_{\mathrm{CE}}$ and $P_{\mathrm{LP}}$ decribe the powers consumed for the processes of channel estimation process and linear processing. Actually, each term depends on the system parameters. Especially, we have that $P_{\mathrm{TC}}=N P_{\mathrm{AP}}+P_{\mathrm{LO}}+K P_{\mathrm{UE}}$, where  $P_{\mathrm{AP}}$, $P_{\mathrm{LO}}$, and $P_{\mathrm{UE}}$
are the powers per AP antenna, AP local oscillator (LO), and the power per user antenna. Moreover, we  have $P_{\mathrm{C-BC}}=B_{\mathrm{w}} \mathrm{ASE} \left( P_{\mathrm{COD}} + P_{\mathrm{DEC}} +P_{\mathrm{BT}} \right) $, where the terms from left to right denote the bandwidth, the powers for data coding  and decoding as well as well as the total power for the backhaul traffic. Regarding the computation of $P_{\mathrm{CE}}$, we have that the MMSE estimation involves $N \tau_{\mathrm{d}}$ and $N$ operations for the calculations of $\bpsi_{k}^{\H}\tilde{\by}_{m}^{\mathrm{tr}}$ and $ \hat{\bh}_{mk}$ in~\eqref{eq:Ypt3} and~\eqref{estimatedChannel1}, respectively. In total, the MMSE estimation requires $K N(\tau_{\mathrm{tr}}  +1)$ operations needing 3 flops per operation with AP computational efficiency $\al_{\mathrm{eff}}$. Given that this procedure takes $\frac{B_{\mathrm{w}}}{\tau_{\mathrm{c}}}$ coherence blocks per second  and $\tau_{\mathrm{tr}}= \frac{K}{\zeta}$, we have
\begin{align}
 P_{\mathrm{CE}}=\frac{3}{L_{\mathrm{AP}}}\frac{B_{\mathrm{w}}}{\tau_{\mathrm{c}}}K N ( \frac{K}{\zeta}  +1).
\end{align}
The linear processing power $P_{\mathrm{LP}}$ is a result of the powers consumed by precoding/transmitting the data and   computation of the precoder, i.e., $P_{\mathrm{LP_{\mathrm{t}}}}$ and $P_{\mathrm{LP_{\mathrm{p}}}}$, respectively. Hence, we have 
\begin{align}
 P_{\mathrm{LP}}=P_{\mathrm{LP_{\mathrm{t}}}}+P_{\mathrm{LP_{\mathrm{p}}}},
\end{align}
where $P_{\mathrm{LP_{\mathrm{t}}}}=\frac{3}{L_{\mathrm{AP}}}\frac{B_{\mathrm{w}}}{\tau_{\mathrm{c}}}K N \xi  (\tau_{\mathrm{c}}  - \tau_{\mathrm{tr}})$ with $\tau_{\mathrm{tr}}=\frac{K}{\zeta}$,  and the power consumed by the conjugate beamformer is given by~\cite{massivemimobook,Pizzo2018} as $P_{\mathrm{LP_{\mathrm{p}}}}=\frac{B_{\mathrm{w} }K}{7\tau_{\mathrm{c}} L_{\mathrm{AP}}}$. Substituting~\eqref{Ptx1} and the power expressions in~\eqref{Pcpc} into~\eqref{APC1}, we conclude the proof.
\end{appendices}

\bibliographystyle{IEEEtran}

\bibliography{mybib}
\end{document}